\documentclass[12pt]{article}

\usepackage{amssymb,amsfonts,amsmath,amsthm,lineno,mathrsfs}
\usepackage{centernot}

\usepackage[T1]{fontenc}
\usepackage[cp1250]{inputenc}
\usepackage{cite}
\usepackage[a4paper, total={16cm, 25cm}]{geometry}
\usepackage{enumitem}
\usepackage{hyperref}

\usepackage{tikz}

\usetikzlibrary{arrows.meta}

\tikzstyle{axis arrow} = [-{Straight Barb[scale=1.2]},line width=1]

\tikzstyle{pointing arrow} = [-{Straight Barb[scale=1.2]},line width=1]

\tikzstyle{axis line} = [line width=1]

\newcommand{\Ar}{A^{\mathrm{ret}}}
\newcommand{\Aa}{A^{\mathrm{adv}}}
\newcommand{\Ac}{A^{\mathrm{Coul}}}
\newcommand{\Ai}{A^{\mathrm{in}}}
\newcommand{\Ao}{A^{\mathrm{out}}}
\newcommand{\ba}{\overline}
\newcommand{\be}{\begin{equation}}
\newcommand{\con}{\mathrm{const}}
\newcommand{\Co}{\mathcal{C}}
\newcommand{\cc}{\mathrm{compl.conj.}}
\newcommand{\ee}{\end{equation}}
\newcommand{\I}{1}
\newcommand{\mR}{\mathbb{R}}
\newcommand{\lb}{\left[}
\newcommand{\rb}{\right]}
\newcommand{\lp}{\left(}
\newcommand{\rp}{\right)}
\newcommand{\s}{\!\cdot\!}
\newcommand{\ac}{\overline{\alpha}}
\newcommand{\D}{\partial}
\newcommand{\du}{d^2 u}
\newcommand{\dl}{d^2 l}

\newcommand{\mrm}{\mathrm}
\newcommand{\e}{\epsilon}
\newcommand{\el}{\varphi_{AB}}
\newcommand{\ex}{\varphi_{AB}(x)}
\newcommand{\elx}{\varphi_{AB}(x)x^B_{A'}}
\newcommand{\ewl}{\partial_\lambda}
\newcommand{\ewo}{\partial_{\lambda_1}}
\newcommand{\ewt}{\partial_{\lambda_2}}
\newcommand{\fb}{\overline{f}}
\newcommand{\f}{\varphi}
\newcommand{\fa}{\varrho}
\newcommand{\fac}{\overline{\varrho}}
\newcommand{\fc}{\overline{\varphi}}
\newcommand{\xc}{\overline{\xi}}
\newcommand{\F}{\Phi}

\newcommand{\g}{\gamma}
\newcommand{\G}{\Gamma}

\newcommand{\io}{\iota}
\newcommand{\z}{\zeta}
\newcommand{\m}{d\mu}
\newcommand{\n}{\nabla}
\newcommand{\oc}{\bar{o}}
\newcommand{\w}{\omega}
\newcommand{\W}{\Omega}
\newcommand{\p}{\psi}
\newcommand{\bp}{\overline\psi}
\newcommand{\ch}{\chi}
\newcommand{\cho}{\dot{\chi}}
\newcommand{\cht}{\ddot{\chi}}

\newcommand{\bc}{\overline\chi}
\newcommand{\la}{\lambda}
\newcommand{\ov}{\overline}
\newcommand{\T}{\mathrm{T}}

\newcommand{\Hc}{\mathcal{H}}
\newcommand{\Sc}{\mathcal{S}}
\newcommand{\si}{\sigma}
\newcommand{\dsp}{\displaystyle}

\DeclareMathOperator{\sgn}{sgn}

\newtheorem{pr}{Proposition}
\newtheorem{lem}[pr]{Lemma}
\newtheorem{col}[pr]{Corollary}

\title{Long-range effects in asymptotic fields and angular momentum of
classical field electrodynamics\,$^\dagger$}
\author{$\mathrm{A{\scriptstyle NDRZEJ}~H{\scriptstyle ERDEGEN}}$\thanks{Alexander von Humboldt Fellow; on leave of absence from:
Institute of Physics, Jagiellonian University, Reymonta 4, 30-059 Cracow,
Poland\newline
\hspace*{1.4em}$^\dagger$\,{\bf  This article has been published in  JMP 36, 4044 (1995) (preceded by preprint \mbox{DESY 95-035}).}}\\
{\it II. Institute of Theoretical Physics, Hamburg University,}\\
{\it Luruper Chaussee 149, 22761 Hamburg, Germany}}
\date{}

\begin{document}
\maketitle

\begin{abstract}
\noindent
Asymptotic properties of classical field electrodynamics are considered.
Special attention is paid to the long-range structure of
the electromagnetic field.
It is shown that conserved Poincar\'e quantities may be expressed in terms
of the asymptotic fields. Long-range variables are shown to be responsible
for an angular momentum contribution which mixes Coulomb and infrared
free field characteristics; otherwise angular momentum and energy-momentum
separate into electromagnetic and matter fields contributions.

\vskip .3in
\noindent
PACS numbers: 03.50.De, 11.10.Jj, 11.30.Cp, 03.70.+k
\end{abstract}

\renewcommand{\thesection}{\arabic{section}}
\renewcommand{\theequation}{\arabic{section}.\arabic{equation}}
\renewcommand{\thepr}{\arabic{section}.\arabic{pr}}

\section{Introduction}
\label{sint}

It is well-known that the long-range character of electromagnetic field
causes certain peculiarities in quantum electrodynamics. Among them the
infraparticle problem and breaking of the Lorentz symmetry are the
most spectacular ones, for a review see the book by Haag \cite{haa} and
an article by Morchio and Strochci \cite{mor}. These properties can be
traced back, as shown most clearly by Buchholz \cite{buc86}, to the fact
provable within the standard system of ideas on properties of
quantum electrodynamics that the flux of electromagnetic field at
spacelike infinity is an essentially classical variable supplying a label
for uncountably many superselection sectors \cite{buc82}. Whether these are
ultimate features of the quantum theory of electromagnetic interaction
or artifacts due to our insufficient understanding of its algebraic structure
is in our opinion an open question as long as we lack consistent,
complete QED beyond Feynman graphs. Doubts about completeness of the
present-state knowledge of the long-range structure can also be raised on
grounds that it tells nothing about the quantization of charge or the
magnitude of the fine-structure constant; see works by Staruszkiewicz on this
point \cite{sta89,sta92}.

In the present work we try to better understand the long-range structure of
electrodynamics in classical field theory. We believe that in this
way one can gain new insights into the quantum case as well.
The domain in which the classical
structure is most likely to be of some
relevance for the quantum case is the asymptotic region. Rigorous
results on the asymptotics of electromagnetic field are presented in Section
\ref{selm} and on the asymptotics of Dirac field in Section \ref{sdir}.
The results are relevant for the interacting theory, as argued in Sections
\ref{selm} and \ref{stot}. In that case some
additional assumptions are made
which seem plausible, but remain unproved. Our main objective, when
discussing the asymptotic fields, is the description of the specific way
how matter and radiation separate in the asymptotic regions.
In this respect the approach of the
present paper differs from that of Flato, Simon and Taflin, who have
recently reported rigorous results on Cauchy problem and scattering states
in classical Maxwell-Dirac theory \cite{fla};
see also a comment in Section \ref{stot}.
Using results on asymptotic fields we express energy-momentum and angular
momentum of the system in terms of those fields.

We stress that our aim is not a purely
mathematical study in classical field theory. Rather, with quantization
in mind, we try to get a reasonably well-founded notion of
the asymptotic structure of fields and conserved Poincar\'e-quantities.

Some of the results described in the present work
were reported earlier in a letter
\cite{her93}. Quantization of the long-range variables within this approach in
a kind of "adiabatic approximation" was discussed in \cite{her}.

Throughout the article we use the abstract index notation \cite{pen}, in which
the index of a geometrical object rather indicates its type, then being a
set of numbers. This interpretation of indices is especially convenient
when spinors are introduced: the two-valence mixed spinors $\rho^{AA'}$
are objects of the same type as complex vectors in Minkowski space,
the respective structures being isomorphic. One identifies, accordingly,
the compound index $AA'$ with the spacetime index $a$ ($BB'$ with $b$ and
so on). We write therefore $\rho^a = \rho^{AA'}$ what in more traditional
notation would be written as $\rho^a = \sigma^a_{AA'}\rho^{AA'}$,
where $\sigma^a_{AA'}$ are the Infeld - van der Waerden symbols giving a
concrete realization of the isomorphism. In this notation the metric tensor
is $g_{ab} = \e_{AB}\e_{A'B'}$, where $\e_{AB}$ is the fixed antisymmetric
spinor (and $\e_{A'B'}\equiv\ba{\e}_{A'B'}$). The correspondence between
an antisymmetric tensor $F_{ab}$ and an equivalent symmetric spinor
$\f_{AB}$ has the form $F_{ab} = \f_{AB}\e_{A'B'} + \fc_{A'B'}\e_{AB}$.
For the null vector of the spinor $o_{A}$ we use fixed notation $l_a =
o_{A}\oc_{A'}$, and for the spinor $\xi_{A}$ respectively $u_a =
\xi_{A}\xc_{A'}$. If below the spinor index in $\oc_{A'}$, $\xc_{A'}$ is
not suppressed and there is no danger of confusion, the bar sign will be
omitted.

\section{Null asymptotics of the electromagnetic field}
\label{selm}

In this section we describe some null asymptotic properties of the
electromagnetic fields. Much of the material is not new, the null infinity
methods being the standard tool in the relativity theory. However, we do not
use the Penrose's conformal compactification, as employed in similar context
in \cite{bra} and \cite{ash}, and use an explicitly Lorentz-covariant
description in terms of homogeneous functions.
Moreover, we describe some global properties in Minkowski space,
which are needed in the discussion of Lorentz generators. The reason for
avoiding the conformal compactification is that it contracts the timelike
past and future infinity to points. This does not seem a natural setting
for the description of massive asymptotic fields living there, which is our
concern in Section \ref{sdir}.

Let us fix the origin in the affine Minkowski space and denote by $x$
a~general point-vector. Let $A(x)$ be a continuous field and suppose it has
well defined asymptotics ${\dsp \lim_{R\to\infty} R\, A(x + Rl)
\equiv b(x, l)}$ for every point $x$ and null vector $l$ (vector and
spinor indices will be often suppressed if no ambiguity arises). $b(x,l)$
is a homogeneous function of degree $-1$ in $l$. Suppose now that $y$ is
a vector lying in the hyperplane $y\s l =0$. If $y\propto l$ then obviously
$b(x+y, l)=b(x, l)$. If $y\not\propto l$, then it is spacelike, and there
always exists a null vector $n$ such that $n\s y =0$, $n\s l =1$.
Then ${\dsp l+ \frac{y}{R} - \frac{y^2}{2R^2}\, n}$ is a future null
vector and
\[
 b(x+y, l) = \lim_{R\to\infty} R\,A\Big(x + \frac{y^2}{2R}\, n
 + R\Big(l+ \frac{y}{R} -\frac{y^2}{2R^2}\, n\Big)\Big)\,.
\]
Therefore, if $A(x)$ is sufficiently regular, one should expect
that again\\ \mbox{$b(x+y, l)=b(x, l)$}, for all $y\s l=0$. This means that
\be
\lim_{R\to\infty} R\, A(x+ Rl)=\ch(x\s l, l)\, ,
\label{asnull1}
\ee
where $\ch(s, l)$ is a homogeneous function of degree $-1$: $\ch(\kappa s, \kappa l)=
\kappa^{-1}\ch(s, l)$. We shall show that (\ref{asnull1}) is indeed satisfied
for a large class, concerning us here, of solutions of the wave equation
(both homogeneous and inhomogeneous).
Instead of null vectors $l_a = o_A o_{A'}$ we shall use
as independent variables the spinors $o$ and $\oc$, adding further conditions
of invariance under the change of the overall spinor phase factor. Thus
$\ch(s, o, \oc )$ will satisfy
\be
\ch(\alpha\ac s, \alpha o, \ac\oc ) = (\alpha\ac)^{-1} \ch(s, o, \oc )
\label{scal1}
\ee
for any complex number $\alpha\neq 0$. The former notation $\ch(s, l)$ will be used
for functions invariant under the overall spinor phase factor change
as a~shorthand.

Let $A(x)\in C^2$ be a global solution of the wave equation
\be
\Box A(x) = 0\, .
\label{wave}
\ee
By the classical Kirchhoff integral formula (see e.g. \cite{pen})
the field $A(x)$ inside the
future lightcone may be recovered from its values on the cone itself.
If these values are represented with the use of a homogeneous function
$\eta(p, l)$ according to
\be\label{cone}
A(R\, l) = R^{-1}\eta\left(R^{-1}, l\right)\, ,\qquad
\eta(\kappa p, \kappa l) = \kappa^{-1}\eta(p, l)\, ,
\ee
then the formula takes on an
especially simple form
\be
A(x) = - \frac{1}{\pi x^2}\int \dot{\eta}\left( 2\frac{x\s u}{x^2}, u\right)\,
\du\, ,
\label{kir}
\ee
where dot over $\eta$ denotes the derivative with respect to the first
argument and $\du$ is the standard invariant measure on the set of null
directions discussed in \ref{apa}. From the homogeneity
property of $\eta$ it follows that the integrand is a homogeneous
of degree $-2$ function of $u$, which is the condition for the applicability of
$\du$. Suppose now, that $RA(Rl)$ has a limit for $R\to\infty$ for all $l$ and
that this limit is achieved without sharp oscillations, which can be expressed
as $\D_R \lp R\, A(Rl)\rp\sim R^{-1-\e}$ for some $\e >0$. Then $A(x)$ has the
anticipated asymptotic behaviour in the whole future lightcone, and moreover
a fall-off property of the asymptotic is implied. More precisely, we have the
following proposition (a $t$-gauge is a scaling of the spinor $o$ for which
$t\s l = 1$, $t^a$ being a timelike unit vector; see \ref{apa}).

\begin{pr}
\label{asnull2}
If $A(x)\in C^2$ is a solution of eq.(\ref{wave}) inside the future lightcone
with the data on the cone given by (\ref{cone}), with $\eta$ in the
$t$-gauge satisfying the bound
$$
\left|\dot{\eta}(p, l)\right| < \frac{\con}{p^{1-\e}}
$$
when
$0<p<p_t$ (for some $\e>0$, $p_t>0$),
then for all $x$ inside the future lightcone the asymptotics
(\ref{asnull1}) holds with
\be
\ch(s, l)= - {1\over 2\pi s}\int \dot{\eta}\Big( {l\s u\over s}, u\Big)\,
\du\, .
\label{fichi}
\ee
$\ch(s, l)\equiv \ch(s, o, \oc )$ has the scaling property (\ref{scal1})
and falls off according to
\be
\left|\ch(s, l)\right| < \frac{\con}{s^\e}
\label{falloff1}
\ee
for ${\dsp s>s_t\equiv {2\over  p_t}}$ in the $t$-gauge.
\end{pr}

We note that the form of the bounds on homogeneous functions as those
appearing in this proposition (and in what follows) is independent of the
choice of the vector $t$ (gauge-independent),
only the bounding constants and $p_t$ (and $s_t$) do change. This is easily
seen
with the use of the inequalities $t\s l\leq e^{\dsp \psi}\,
\tilde{t}\s l$ and $\tilde{t}\s l\leq e^{\dsp \psi}\,
t\s l$ for any null vector $l$ and any two unit, timelike,
future-pointing vectors $t$ and $\tilde{t}$, where
$t\s\tilde{t} = \cosh \psi$.

\begin{proof}
Fix $x^a=\la z^a$, $z^2 = 1$, $z^0>0$, and choose $l$ and $u$
in $z$-gauge. Parametrize $u$ by $(\rho, \f)$ as in \ref{apa} (with $z$
playing the role of the time-vector) and change the $\rho$ variable to
$\rho_0 >0$ by $\dsp \rho_0^2=\frac{2R\rho^2+\la}{2R+\la}$.
Then, by (\ref{kir}),
$$
RA(x +Rl) = \frac{-1}{2\pi\la} \int_0^1 2d\rho_0^2 \int d\f\,
\dot{\eta}\left(\frac{2}{\la}\rho_0^2, u(\rho(\rho_0), \f)\right)\,
\theta\left(\rho_0^2-\frac{\la}{\la+2R}\right)\, .
$$
The integrand is bounded in module by
$\dsp
 \theta\Big(p_z -\frac{2}{\la}\rho_0^2\Big)\frac{\con}{(\rho_0^2)^{1-\e}} +
\con\,\theta\Big(\frac{2}{\la}\rho_0^2 - p_z\Big),
$
hence by the Lebesgue theorem
$$
\lim_{R\to\infty} RA(x+Rl) = \frac{-1}{2\pi\la}\int_0^1 2d\rho^2\int d\f\,
\dot{\eta}\left(\frac{2}{\la}\rho^2, u(\rho, \f)\right)\, ,
$$
which is (\ref{asnull1}) with $\ch$ (\ref{fichi}) in the $z$-gauge and
$(\rho, \f)$-parametrization. The bound for
$\ch$ is easily obtained.
\end{proof}

The next proposition gives the field itself from its asymptotic, by a slightly
strengthened fall-off condition (no sharp oscillations).

\begin{pr}
\label{freeas1}
Let $\ch(s, l)$ and its derivatives with respect to $s$ of up to the third
order be continuous functions of $s$ and $l$ for $s\in\mR$. Suppose
$\ch(s, l)$ and $\cho(s, l)$ satisfy respectively (\ref{falloff1}) and
\be
\left|\cho(s, l)\right| < \frac{\con}{s^{1 + \e}}
\label{falloff2}
\ee
for $s>s_t>0$ in the $t$-gauge.

Then $\ch(s, l)$ is the asymptotic (\ref{asnull1}) of the field
\be
A(x) = - {1\over 2\pi}\int \cho(x\s l, l)\,\dl\, ,
\label{freeas2}
\ee
which satisfies the wave equation. For a given $x$ there is
$|R A(x + Rl)| < \con$ for all $R\geq 0$ and $l$ in the $t$-gauge.

If in addition we demand that also $\cht(s, l)$ satisfies (\ref{falloff2}),
then uniqueness of $A(x)$ with the given asymptotic is guaranteed.
\end{pr}

\begin{proof}
The wave equation is obviously satisfied. Further\\
\mbox{$R A(x + Rl) = {-1\over 2\pi}\int\cho(x\s u + Rl\s u, u)\, \du$.}
Choose $l$ and $u$ in the $t$-gauge, use $(\rho, \f)$-para\-metrization for
$u$ as in \ref{apa} and replace the $\rho$ variable by
\mbox{$\beta = 2R\rho^2$}. Then
\[
 R A(x + Rl) = {-1\over 2\pi}
\int_0^{2\pi}d\f\int_0^{2R}d\b\,\cho(x\s u +\beta, u),
\
\text{where}\
u=u(\rho , \f ) = u\left( \sqrt{\beta/2R}, \f\right).
\]
For \mbox{$\beta>s_t+|x^0| + |\vec{x}|$}
the condition (\ref{falloff2}) implies
\[
 |\cho(x\s u + \beta, u)|
< \con (\beta - |x^0| - |\vec{x}|)^{-1-\e}.
\]
For $0\leq\beta\leq s_t + |x^0| + |\vec{x}|$ there is
\[
 -(|x^0| + |\vec{x}|)\leq x\s u + \beta\leq 2(|x^0| + |\vec{x}|) + s_t,
\]
so, for fixed $x$, by continuity
$|\cho(x\s u + \beta, u)| <\con$ in this case. The asymptotic (\ref{asnull1}) and
the bound follow now easily. If (\ref{falloff2}) is assumed for $\cht(s, l)$,
then also the field
$\dsp \n_aA(x) = -{1\over 2\pi} \int l_a\cht(x\s l, l)\, \dl$
has similar asymptotic properties.
By the Kirchhoff formula $A(x)$ can be uniquely recovered from the values of
$A$ on any past-directed lightcone, such that $x$ lies inside the cone;
the formula involves the field itself on the cone and its derivative along the
generating lines of the cone. If one tends with the vertex of the cone to
the future timelike infinity, then the integrands tend to respective null
asymptotics. The dominated convergence given by the proposition gives then
(\ref{freeas2}) as the limit of the Kirchhoff formula, which implies
uniqueness.
\end{proof}

Up to now we have considered the null asymptotic in the future direction only.
In exactly the same way the past null asymptotic can be considered.
Proposition \ref{freeas1} holds again, with
(\ref{falloff1}), (\ref{falloff2}), (\ref{asnull1}) and
(\ref{freeas2}) replaced respectively by
\be
\left|\ch'(s, l)\right| < \frac{\con}{|s|^\e}
\label{falloff3}
\ee
and
\be
\left|\cho'(s, l)\right| < \frac{\con}{|s|^{1 + \e}}
\label{falloff4}
\ee
for $s<s'_t<0$ in the $t$-gauge,
\be
\lim_{R\to\infty} R\, A(x- Rl)=\ch'(x\s l, l)\, ,
\label{asnull3}
\ee
\be
A(x) = {1\over 2\pi}\int \cho'(x\s l, l)\,\dl\, .
\label{freeas3}
\ee
The null asymptotics $\ch$ and $\ch'$ are not independent, as for any $x$
the representations (\ref{freeas2}) and (\ref{freeas3}) have to agree:
\be
\int \Sigma(x\s l, l)\, \dl =0\, ,
\label{cons1}
\ee
where
\be
\Sigma(s, l) = \cho(s, l) + \cho'(s, l).
\label{cons2}
\ee
\begin{lem}
\label{cons3}
If $\Sigma(s, l)$ is continuous, satisfies (\ref{cons1}),  and
$|\Sigma(s, l)|$ is bounded by some polynomial in $s$ (in some $t$-gauge),
then
\[
 \Sigma(s, l) = \sum_{k=0}^N s^k\Sigma_k(l)\quad \text{with}\quad
 N<\infty\quad \text{and}\quad
 \int l_{a_1}\dots l_{a_k}\Sigma_k(l)\, \dl  = 0.
\]
\end{lem}

Before giving a proof we fix our conventions for the Fourier
transformations. For $f(x)$ a function of the spacetime point and
$g(s, o, \oc )$ a function of the spinor $o$ and a real variable $s$ we
denote
\be
\hat{f}(p) = {1\over 2\pi}\int f(x) e^{\dsp i\, p\s x}\, d^4x\, ,
\label{confour1}
\ee
\be
\tilde{g}(\w, o, \oc) = {1\over 2\pi}\int g(s, o, \oc)
e^{\dsp i\, \w s}\, ds\, .
\label{confour2}
\ee
If $g(\alpha\ac s, \alpha o, \ac\oc)= \alpha^p\ac^q g(s, o, \oc)$ then
$\dsp \tilde{g}\left({\w\over\alpha\ac}, \alpha o, \ac\oc\right) =
\alpha^{p+1}\ac^{q+1}\tilde{g}(\w, o, \oc)$.

\begin{proof}[Proof of the lemma]
We integrate the condition (\ref{cons1}) with a function
of fast decrease $f(x)$. The result can be rewritten as
$\dsp \int \tilde{\Sigma}(\w, l)\ov{\hat{f}(\w l)}\, d\w\, \dl =0$.
Fix a $t$-gauge and assume that $\hat{f}(\w l) = g(\w) h(l)$; this can be
extended to a~Schwartz function $\hat{f}(p)$ if $h(l)$ is infinitely
differentiable and $g(\w)$ is a~Schwartz function vanishing in some
neighbourhood of $\w =0$. Going over all possible $h(l)$ we obtain
$\dsp \int \tilde{\Sigma}(\w, l)\ov{\tilde{g}(\w)}\, d\w = 0$
for all $l$. Thus, for every $l$, $\tilde{\Sigma}(\w, l)$ is a
distribution concentrated in $\w =0$, hence a finite linear combination of
derivatives of $\delta(\w)$. By polynomial boundedness the supremum over $l$ of
the degree of the highest derivative of $\delta(\w)$ is finite, hence $N<\infty$.
Inserting the expansion into (\ref{cons1}) one obtains
the constraints on $\Sigma_k(l)$.
\end{proof}

The result of the lemma gives via (\ref{cons2}) the relation between the
future and the past null-asymptotics. We add now the physical condition
that the energy of the field be finite. It will be seen below, that it leads
for the electromagnetic field to the condition which corresponds here to
the integrability of $\left|\cho(s, l)\right|^2$ over all $s\in
(-\infty, +\infty)$. By the result of the lemma and the fall-off
conditions for $\cho(s, l)$ and $\cho'(s, l)$ it follows now that $\Sigma(s, l)$
vanishes identically. Thus we have
$$
\cho(s, l) + \cho'(s, l) = 0\, .
\label{cons4}
$$
This implies that both $|\cho(s, l)|$ and $|\cho'(s, l)|$ satisfy
both fall-off conditions (\ref{falloff2}) and (\ref{falloff4}).
This implies also that there exist limits
\[
 \ch(-\infty, l)\equiv\lim_{s\to\ -\infty} \ch(s, l),\qquad
 \ch'(+\infty, l)\equiv\lim_{s\to\ +\infty} \ch'(s, l)
\]
and
\be
\ch(s, l) + \ch'(s, l) = \ch(-\infty, l) = \ch'(+\infty, l)\, .
\label{cons5}
\ee
If we think of $A(x)$ as an analog of the electromagnetic potential, or in
fact a component of the latter, then the fields with nonvanishing
$\ch(-\infty, l)$ are exactly those infrared-singular in the usual sense
(as observed in \cite{ash}). This is easily seen, when the connection between
(\ref{freeas2}) and the usual Fourier representation is clarified. This is
simply achieved if $\cho(x\s l, l)$ in (\ref{freeas2}) is represented by its
transform (\ref{confour2}). If we write the Fourier representation of the
field $A(x)$ as
\be
A(x) = {1\over \pi} \int a(k)\delta(k^2)\e(k^0)\, e^{\dsp -i x\s k}\,
d^4k\, ,
\label{Four}
\ee
then
\be
a(\w l) = - {\tilde{\cho}(\w, l)\over \w}\, .
\label{afi}
\ee
But $\dsp \tilde{\cho}(0, l)= -{1\over 2\pi}\ch(-\infty, l)$,
as easily seen from
(\ref{confour2}). If this does not vanish, then $a(\w l)\sim \w^{-1}$ at the
origin. We note that the function $\ch(-\infty, l)$ is not only (as $\ch(s, l)$)
Lorentz-frame independent, but also independent of the choice of the origin in
Minkowski space. It describes uniquely the spacelike asymptotic of the field
$A(x)$:
\be
\lim_{R\to\infty}R\, A(x + Ry) = {1\over 2\pi}\int\ch(-\infty, l)\delta(y\s l)\,
\dl\, ,
\label{spas}
\ee
where $\delta$ is the Dirac distribution. This is true both point-like in $y$
for $y^2\neq 0$ (for timelike $y$ yielding simply $0$), and distributionally
when integrated with a test function $f(y)$, for any fixed $x$. One proves
this by a method similar to that used in the proof of the Proposition~\ref{freeas1}.

In the next step we want to take into consideration fields with nonvanishing
sources, satisfying equation
\be
\Box A(x) = 4\pi J(x)\, .
\label{sour1}
\ee
Some restrictions on the current density have to be assumed. As the scattering
aspects are those which concern us here, we want the free radiation field
$A^{\mrm{rad}} = \Ar - \Aa$ to fall into the class of fields considered
up to now. The Pauli-Jordan function
$\dsp D(x)={1\over 2\pi}\e(x^0)\delta(x^2)$
can be written in the representation (\ref{freeas2}) as
$$
D(x) = -{1\over 8\pi^2}\int \delta'(x\s l)\, \dl\, .
$$
Using it, we obtain $A^{\mrm{rad}}(x)$ in the representation
(\ref{freeas2}) with the integrand\\ $\cho^{\mrm{rad}}(s, l) =
\dot{c} (s, l)$,
\be
c(s, l) = \int\delta(s - l\s y) J(y)\, d^4y\, .
\label{sour2}
\ee
We assume therefore that this function is well defined, and that
$\dot{c} (s, l)$ satisfies the premises of Proposition \ref{freeas1}.
Note, however, that $c(+\infty, l)$ need not vanish, and the future null
asymptotic of $A^{\mrm{rad}}$ is given by $\ch^{\mrm{rad}}(s, l)=
c(s, l) - c(+\infty, l)$. Suppose further that the support of
the current is bounded in
spacelike directions, that is for every $x$ the set $\{y|\, y^2\leq 0,
J(x + y) \neq 0\}$ is bounded. This condition can be relaxed to some decay
in spacelike directions, but we dot study this problem in detail.
The asymptotics of the retarded and advanced solutions is easily found
\be
\lim_{R\to\infty}R\,\Ar(x + Rl) = \lim_{R\to\infty}R\,\Aa(x - Rl) =
c(x\s l, l)\, .
\label{sour3}
\ee
Combined with the asymptotics of the radiation field this also gives
\begin{equation}\label{sour4}
 \lim_{R\to\infty}R\,\Ar(x - Rl) = c(-\infty, l)\, , \qquad
 \lim_{R\to\infty}R\,\Aa(x + Rl) = c(+\infty, l)\, .
\end{equation}
The incoming and outgoing fields are defined as usual by
$A = \Ai + \Ar = \Ao + \Aa$ and they are assumed to belong to the class of
free fields considered here. Denoting $\ch$ and $\ch'$ the future and past
null asymptotics of $A$ we have the relations
\begin{equation*}
\ch'(s, l) = \ch'^{\mrm{in}}(s, l) + c(-\infty)\, ,\qquad
\ch(s, l) = \ch^{\mrm{out}}(s, l) + c(+\infty)\, .
\end{equation*}
The full relation (\ref{cons5}) is lost now, but it remains true, that
\be
\ch(-\infty, l) = \ch'(+\infty, l)\, .
\label{cons6}
\ee

The extension of the preceding discussion to the case of the electromagnetic
fields involves some physically important modifications. The Maxwell
equations in the spinor form read
\be
\n^A_{B'}\ex = 2\pi J_b(x)
\label{max1}
\ee
with a real conserved current $J_b$. If complex $J_b$ is admitted, its
imaginary part is the magnetic current of the generalized Maxwell equations in
tensor form. To see the physical consequences of the absence of magnetic
currents in the context of asymptotic fields we take this condition only later
into account. We shall see later that in order that the radiated
energy-momentum and angular momentum be well defined not only $\ex$ but
also $\elx$ should have the asymptotic behaviour discussed above.
As the latter field appears repeatedly in the present context it is
convenient to denote
\be
\fa_{AA'}(x)=\f_{AB}(x)x^B_{A'}\, .
\label{fa}
\ee
This field satisfies
\be
\n^A_{B'}\fa_{AA'}(x) = 2\pi J_{BB'}(x)x^B_{A'}\, .
\label{max2}
\ee
From (\ref{max1}) and (\ref{max2}) the inhomogeneous wave equations follow
\begin{align}
 \Box\ex &= 4\pi \n_{AC'}J^{C'}_B(x)\, , \label{box1}\\[1ex]
 \Box\fa_{AA'}(x) &= 4\pi \n_{AC'}\big(J^{C'}_B(x)x^B_{A'}\big)\, .
\label{box2}
\end{align}

In the free field case we demand therefore that
\begin{align*}
 \ex &= -{1\over 2\pi}\int\dot{f}_{AB} (x\s l, o, \oc)\, \dl\, ,\\
 \fa_{AA'}(x) &= -{1\over 2\pi}\int\dot{h}_{AA'} (x\s l, o, \oc)\, \dl\, ,
\end{align*}
where $|\dot{f}_{AB} (s, o, \oc)|$ and $|\dot{h}_{AA'} (s, o, \oc)|$
are bounded by $\con |s|^{-1-\e}$ for large $|s|$ and both
$|f_{AB} (s, o, \oc)|$ and $|h_{AA'} (s, o, \oc)|$ vanish for $s\to +\infty$.
(Here and in what follows such bounds on spinor and tensor functions are to
be understood in some $t$-gauge, component-wise in some Minkowski frame
in which time axis is parallel to $t$, and in the associated spinor frame;
for fixed vector $t$ the bounds do not depend on the choice of the
spacelike frame.)
Setting these formulae into (\ref{max1}) and (\ref{max2}) respectively
(with $J_b=0$) and taking null asymptotics one obtains
$\dsp o^A f_{AB} (s, o, \oc) = 0$,
$\dsp o^A h_{AA'} (s, o, \oc) = 0$, i.e.
$\dsp f_{AB} (s, o, \oc)= o_A o_B f(s, o, \oc)$,
$\dsp h_{AA'} (s, o, \oc)= o_A h_{A'} (s, o, \oc)$.
(These formulae
follow also from the generalized Kirchhoff formula \cite{pen}.) Contracting
the above representation of $\ex$ with $x^B_{A'}$ and using (\ref{inth2})
we have
\begin{align*}
 \elx &= {1\over 2\pi}\int o_A x_{A'B} o^B \dot{f}(x\s l, o, \oc)\,\dl
 = {1\over 2\pi}\int o_A (\D_{A'} - \D'_{A'}) f(x\s l, o, \oc)\,\dl\\
  &=  -{1\over 2\pi}\int o_A \D'_{A'} f(x\s l, o, \oc)\,\dl\, ,
\end{align*}
where $\dsp \D_{A'}\equiv {\D\over \D o^{A'}}$ and
$\dsp \D'_{A'} f(x\s l, o, \oc) \equiv
\D_{A'} f(s, o, \oc)\Big|_{\textstyle s = x\s l}$.

From now on we make a general assumption that the spinor derivatives of the
asymptotics up to the order which will appear in the future considerations
do not spoil the fall-off properties, so that e.g. together with
$\dsp |f(s, o, \oc)|$ also $\dsp |\D_{A'}f(s, o, \oc)|$
falls off as $s^{-\e}$ for
$s\to +\infty$ and with $\dsp |f(s, o, \oc) - f(-\infty, o, \oc)|$
also $\dsp |\D_{A'}[f(s, o, \oc) - f(-\infty, o, \oc)]|
\sim |s|^{-\e}$ for
$s\to -\infty$.
From the homogeneity properties of asymptotics then follows
that the differentiation with respect to $s$ increases the rate of
fall-off by one inverse power of $|s|$, as e.g. \mbox{$\dsp o^{A'}\D_{A'} f
+ s\dot{f} = -f$}.

Comparing now the two above representations of $\dsp \elx$
and using Lemma \ref{cons3} we have
$\D_{A'} f(s, o, \oc) = \dot{h}_{A'} (s, o, \oc)$.
Contracting this with $o^{A'}$ and using homogeneity we get
$\dsp \D_s (sf) = -o^{A'} \dot{h}_{A'}$,
or $\dsp sf = o^{A'} h_{A'} - g$, where
\mbox{$g=g(o, \oc)$} has the homogeneity property
$g(\alpha o, \ac\oc)= \alpha^{-2}g(o, \oc)$.
Differentiation on $\D_{B'}$ yields
$\dsp s\,\dot{h}_{B'} = \D_{B'}(o^{A'} h_{A'}) - \D_{B'} g$.
The l.h. side vanishes
for $|s|\to\infty$, whereas the r.h. side tends to $-\D_{B'} g$ for
$s\to +\infty$ and to $\dsp \D_{B'}(o^{A'} h_{A'}(-\infty) - g)$ for
$s\to -\infty$. Hence
$\dsp \D_{B'} g = \D_{B'}(o^{A'} h_{A'}(-\infty)) = 0$,
which implies by (\ref{swh2}) $\dsp g= o^{A'}h_{A'}(-\infty)=0$.
Therefore $f=s^{-1}o^{A'} h_{A'} \sim |s|^{-1-\e}$ for $|s|\to\infty$,
so there is a unique representation $f=\dot{\z}$ with $\z$ vanishing for
$s\to +\infty$.

Summarizing the free electromagnetic field case we have
\begin{align}
 \ex &= -\frac{1}{2\pi}\int o_A o_B \ddot{\z}(x\s l, o, \oc)\, \dl\, ,
 \label{frel1}\\
 \fa_{AA'}(x) &=
 -\frac{1}{2\pi}\int o_A \D'_{A'}\dot{\z}(x\s l, o, \oc)\, \dl\, ,
 \label{frel2}
\end{align}
\begin{align}
 \lim_{R\to\infty}R\,\el(x + Rl) &= o_A o_B \dot{\z}(x\s l, o, \oc)\, ,
 \label{elas1}\\
 \lim_{R\to\infty}R\,\fa_{AA'}(x+Rl)&= o_A \D'_{A'}\z(x\s l, o, \oc)\, .
\label{elas2}
\end{align}
This class of fields admits a~class of Lorentz-gauge potentials
with properties characterized by Proposition \ref{freeas1}
\be
A_a(x)=-\frac{1}{2\pi}\int \dot{V}_a(x\s l, l)\, \dl\, .
\label{freepot1}
\ee
$V_a(s, l)$ is a real vector function with properties of $\chi(s, l)$ of
Proposition \ref{freeas1} and such that
\be
o_{C'}V^{C'}_A(s, l)=o_A\z(s, o, \oc)\, .
\label{freepot2}
\ee

Turning now to the asymptotics of the retarded and advanced fields
$$
\f^{\mrm{ret,adv}}{}_{AB}(x) = 4\pi \int G^{\mrm{ret,adv}}(x - y)
\n_{AC'} J^{C'}_B(y)\, d^4y
$$
we observe first that the fields $\f^{\mrm{ret}}{}_{AB}(x)x^B_{A'}$ and
$\f^{\mrm{adv}}{}_{AB}(x)x^B_{A'}$ may be obtained as respectively retarded and
advanced solutions of (\ref{box2}). This is seen as follows.
Suppressing the labels "ret" or "adv" we have
\begin{equation*}
 4\pi\int G(x-y) \n_{AC'} \big(J^{C'}_B(y)y^B_{A'}\big)d^4y
 - \ex x^B_{A'}\\
 = 4\pi\int G(z)\n^{(z)}_{AC'} \big(J^{C'}_B(x-z)z^B_{A'} \big)d^4z\,.
\end{equation*}
Using the conservation law of $J_b$ and the rules for transforming spinor
into tensor expressions one has (all differentiations on $z$)
\begin{align*}
 \n_{AC'}&\big(J^{C'}_B(x-z)z^B_{A'}\big)
 = J_a(x-z) + z^B_{A'}\n_{BC'} J^{C'}_A(x-z) \\
 &= \Big( {1\over 2} z\s\n + 1\Big)J_a(x-z) +
\left( z_{[a}\n_{c]} - i e_{acbd}z^b\n^d\right) J^c(x-z)\, .
\end{align*}
As the retarded and advanced Green functions satisfy
$$
(z\s\n + 2)G(z) =0\, ,~~~(z_a\n_b - z_b\n_a)G(z)=0\, ,
$$
the above integral vanishes, which ends the proof of our statement.
This property implies that one can attach the ret/adv labels to
$\fa_{AA'}$ without risk of ambiguity.
The leading asymptotic terms can be now simply represented.
If we denote
\be
c_A (s, o, \oc) =\int \delta (s-x\s l)J^{C'}_A(x)\, d^4x\, o_{C'}\, ,
\label{elce}
\ee
then
\begin{align*}
 \lim_{R\to\infty}R\,\f^{\mrm{ret}}{}_{AB}(x+Rl)
 &= \lim_{R\to\infty}R\,\f^{\mrm{adv}}{}_{AB}(x-Rl)
 = o_{(A} \dot{c}_{B)}(x\s l, o, \oc)\, ,\\[1ex]
 \lim_{R\to\infty}R\,\fa^{\mrm{ret}}{}_{AA'}(x+Rl)
 &=\lim_{R\to\infty}R\,\fa^{\mrm{adv}}{}_{AA'}(x-Rl)
 = \D'_{A'}c_A(x\s l, o, \oc)\, .
\end{align*}
The last two equalities follow from
\begin{align*}
 \D_{A'}c_A&(s, o, \oc)= \int\left\{\delta'(s-x\s l)x^B_{A'} o_B o_{C'} J^{C'}_A(x) - \delta(s-x\s l)J_a\right\}d^4x\\[1ex]
 &=\int\delta(s-x\s l)\left(J_a + x^B_{A'}\n_{BC'}J^{C'}_A\right)d^4x =
 \int\delta(s-x\s l)\n_{AC'}\left(J^{C'}_B(x)x^B_{A'}\right)d^4x\, .
\end{align*}
As in the case of the scalar field the current is assumed such that
$\dot{c}_A(s, o, \oc)$ has the required fall-off properties (implying the
existence of limits $c_A(\pm\infty, o, \oc)$). The formulae analogous to
(\ref{sour4}) are
\begin{align*}
 &\lim_{R\to\infty}R\,\f^{\mrm{ret}}{}_{AB}(x-Rl) =
 \lim_{R\to\infty}R\,\f^{\mrm{adv}}{}_{AB}(x+Rl) = 0\, ,\\[1ex]
 &\lim_{R\to\infty}R\,\fa^{\mrm{ret}}{}_{AA'}(x-Rl) =
 \D'_{A'}c_A(-\infty , o, \oc)\, ,\\[1ex]
 &\lim_{R\to\infty}R\,\fa^{\mrm{adv}}{}_{AA'}(x+Rl) =
 \D'_{A'}c_A(+\infty , o, \oc)\, .
\end{align*}
There are further conditions on $c_A$ following from the conservation of
$J_a$, and from its reality, when this is the case (pure electrodynamics).
From the conservation law we have
$$
0= \int\delta (s-x\s l)\n_aJ^a(x)\, d^4x =
\int\delta' (s-x\s l)J^a(x)\, d^4x\, l_a\, ,
$$
that is $\dot{c}_A o^A =0$ or $c^A(s, o, \oc) o_A = Q = Q_{\mrm{el}}
- i Q_{\mrm{mag}}$. $Q_{\mrm{el}}$ and $Q_{\mrm{mag}}$ are the
electric and the magnetic charge of the field respectively. The last equation
implies also $\dot{c}_A(s, o, \oc)\propto o_A$ and
$\D_{A'}c_A(s, o, \oc)\propto o_A$. To see the consequence of reality of
$J_a$ we choose an arbitrary spinor $\io_A$ complementing $o_A$ to
a~normalized spinor basis $o_A\io^A = 1$ and decompose in the standard
null tetrad \cite{pen}
$$
\int\delta(s-x\s l) J_a(x)\, d^4x = \alpha(s, l)l_a + \beta(s, l)m_a +
\g(s, l)\bar{m}_a + Qn_a\, .
$$
If $J_a$ is real, then $\alpha(s, l)$ and $Q$ are real and $\g(s, l)=
\ov{\beta(s, l)}$. The only condition implied in this case for
$c_A(s, o, \oc) = \beta(s, l) o_A + Q\io_A$ is the reality of $Q$.
Moreover, in that case the retarded (advanced) Lorentz-gauge potentials
$A^{\mrm{ret}}{}_a(x)$ ($A^{\mrm{adv}}{}_a(x)$) have the required null
asymptotic behaviour with asymptotics characterized by
\be
c_a(s, l)=\int\delta(s-x\s l)J_a(x)\, d^4x\, .
\label{rapot}
\ee

We summarize the general field case now easily obtained as a superposition
of a free and the ret/adv fields. The necessary terms of the
electromagnetic field asymptotics are represented with the use of a spinor
function $\z_A(s, o, \oc)$ with the fall-off
\be
\left|\dot{\z}_A(s, o, \oc)\right| < \frac{\con}{|s|^{1+\e}}
\label{zfo}
\ee
for $|s|>s_t>0$,
differential properties as assumed for $\z$ above, homogeneity
\be
\z_A(\alpha\ac s, \alpha o, \ac\oc) = \alpha^{-1} \z_A(s, o, \oc)
\label{zet1}
\ee
and satisfying in addition
\be
\z^A(s, o, \oc) o_A = Q = Q_{\mrm{el}} - i Q_{\mrm{mag}}\, .
\label{zet2}
\ee
Then
\begin{align}
 &\lim_{R\to\infty}R\,\el(x + Rl) = o_A \dot{\z}_B (x\s l, o, \oc)\, ,
 \label{elas3}\\[1ex]
 &\lim_{R\to\infty}R\,\fa_{AA'}(x+Rl)
 = \D'_{A'}\z_A(x\s l, o, \oc)
 \equiv o_A \nu_{A'}(x\s l, o, \oc)\, .
\label{elas4}
\end{align}
The last identity is the definition of $\nu_{A'}$. If $J_a(x)$ is real
than there exists a~class of Lorentz-gauge potentials with null asymptotics
\be
\lim_{R\to\infty}R\, A_a(x + Rl) = V_a(x\s l, l)\, ,
\label{pot1}
\ee
where $V_a(s, l)$ has the properties of $\chi(s, l)$ of the scalar case
and satisfies
\be
o_{C'}V^{C'}_A(s, l)= \z_A(s, o, \oc)\, .
\label{pot2}
\ee

Past null asymptotics
are similarly given by another function $\z'_A(s, o, \oc)$ with the same
properties. As in the scalar case there is
\be
\z_A(-\infty, o, \oc) = \z'_A(+\infty, o, \oc)\, .
\label{zet3}
\ee
The future null asymptotic of the free outgoing field is given by
\be
\z_A(s, o, \oc) - \z_A(+\infty, o, \oc)
\equiv o_A \z^{\mrm{out}}(s, o, \oc)\, ,
\label{zet4}
\ee
which is the definition of $\z^{\mrm{out}}(s, o, \oc)$ at the same time.
Similarly, the past null asymptotic of the incoming field is supplied by
\be
\z'_A(s, o, \oc) - \z'_A(-\infty, o, \oc)
\equiv o_A \z'^{\mrm{in}}(s, o, \oc)\, .
\label{zet5}
\ee
One observes that the asymptotically relevant (needed for determination of the
radiated angular momentum, as will be seen later)
information on the asymptotics of
the electromagnetic field is not fully contained in the free outgoing or
incoming fields.
The remaining terms
$$
\z_A(+\infty, o, \oc) = c_A(+\infty, o, \oc)
= \lim_{s\to +\infty}\int \delta (s-x\s l)J^{C'}_A(x)\, d^4x\, o_{C'}
$$
and
$$
\z'_A(-\infty, o, \oc) = c_A(-\infty, o, \oc)
= \lim_{s\to -\infty}\int \delta (s-x\s l)J^{C'}_A(x)\, d^4x\, o_{C'}
$$
are connected with the Coulomb fields of the outgoing and incoming
currents respectively.

The physical significance of the limit values $\z_A(\pm\infty, o, \oc)$ and
$\z'_A(\pm\infty, o, \oc)$ is revealed by considering the spacelike limit
of the electromagnetic field. For a free field (\ref{frel1}) one obtains
by the method used already in the scalar case
\be
\lim_{R\to\infty} R^2 \f^{\mrm{free}}{}_{AB}(a+Ry) = {1\over 2\pi}
\int \delta'(y\s l) o_A o_B\z(-\infty, o, \oc)\, \dl
\label{espas1}
\ee
for any point $a$ and spacelike vector $y$. For a general field we use a trick,
which will be useful also in the next section. Decompose the field $\el$
into the retarded and free outgoing fields $\el = \f^{\mrm{ret}}{}_{AB}[J] +
\f^{\mrm{out}}{}_{AB}$, where the source $J_a$, according to our earlier
assumptions, has finite extension in spacelike directions.
Choose an arbitrary point $a$ and a
time-axis through $a$ in the direction of a unit timelike, future-pointing
vector $t$. For real positive $c$ denote
by $\Co_s^{\mrm{fut}}(-c)$ the solid future lightcone with the vertex
in $a - ct$, and by $\Co_s^{\mrm{past}}(c)$ the solid past lightcone
with the vertex in $a + ct$. Choose $c$ such that $J_b =0$ in
$R(c) \equiv\mathbf{\mrm{M}} \setminus \left\{\Co_s^{\mrm{fut}}(-c)\cup
\Co_s^{\mrm{past}}(c)\right\}$. The retarded field is not influenced in
$R(c)$ by the values of the current in $\Co_s^{\mrm{fut}}(-c)$. We make
advantage of this fact to replace $J_b$ by a current $J'_b$ which is
identical with $J_b$ in the past of $\Co_s^{\mrm{fut}}(-c)$ but to
the future of $\Co_s^{\mrm{past}}(c)$ represents a point charge $Q$
(possibly both electric and magnetic) sitting on the time-axis.
This is always possible, since the charge is the only characteristic of
a current which cannot be deformed without violation of the continuity
equation. Thus in the region $R(c)$ we can write
$\el = \f^{\mrm{ret}}{}_{AB}[J'] + \f^{\mrm{out}}{}_{AB}$.
However, if $J_b$ belongs to the class of admitted
currents, so does $J'_b$, and the radiated field $\f^{\mrm{rad}}{}_{AB}[J'] =
\f^{\mrm{ret}}{}_{AB}[J'] - \f^{\mrm{adv}}{}_{AB}[J']$ is an admissible free field.
On the other hand $\f^{\mrm{adv}}{}_{AB}[J']$ is identical in $R(c)$ with
the Coulomb field of a point charge $Q$ sitting on the time-axis.
Summarizing, the field $\el$ can be represented in $R(c)$ by
$\el = \f^Q{}_{AB} + \f^{\mrm{free}}{}_{AB}$, where $\f^Q{}_{AB}$ represents this Coulomb
field and $\f^{\mrm{free}}{}_{AB}$ is a free field. The region $R(c)$ is large
enough for this representation to be used for determination of
(i) future null asymptotics for $s<s_1$, for some $s_1$, in $t$-gauge;
(ii) past null asymptotics for $s>s_2$, for some $s_2$, in $t$-gauge;
(iii) spacelike asymptotics from the point $a$.

For the (generalized -- with possible magnetic charge) Coulomb field
\be
\f^Q{}_{AB}(x) =
{Q\over \left(\left[ (x-a)\s t\right]^2  - (x-a)^2\right)^{3/2} }\,
t^{C'}_{(A} (x-a)_{B)C'}
\label{Coul}
\ee
one has
$$
\lim_{R\to\infty} R^2 \f^Q{}_{AB}(a+Ry) =
{Q\over \left( (y\s t)^2  - y^2\right)^{3/2} }\, t^{C'}_{(A} y_{B)C'}\, .
$$
The null asymptotics are most easily found with the use of (\ref{elce})
\be
\z^Q{}_A(s, o, \oc) = \z'^Q{}_A(s, o, \oc) = {Q\over t\s l}\, t^{C'}_A o_{C'}\, .
\label{nullC}
\ee
Using the identity $\dsp \int\delta(y\s l)\,{\dl\over t\s l} =
{2\pi\over \left( (y\s t)^2  - y^2\right)^{1/2}}$ we find the relation
$$
{1\over 2\pi}\int\delta'(y\s l) o_{(A}\z^Q_{B)} (-\infty, o, \oc)\, \dl =
\lim_{R\to\infty} R^2 \f^Q{}_{AB} (a+Ry)\, .
$$

Comparison with the free field case and the use of the trick described above
allow us to write in general case for any point $a$ and spacelike
\mbox{vector $y$}
\be
\lim_{R\to\infty} R^2 \el(a+Ry) = {1\over 2\pi}
\int \delta'(y\s l) o_{(A} \z_{B)}(-\infty, o, \oc)\, \dl\, .
\label{espas2}
\ee
Formulae (\ref{espas1}) and (\ref{espas2}) furnish the required interpretation
of the limit values of asymptotics. The limit value $\z_A(-\infty, o, \oc) =
\z'_A(+\infty, o, \oc)$ describes the long-range degrees of freedom of the
total field; $\z^{\mrm{out}}(-\infty, o, \oc)$ and
$\z'^{\mrm{in}}(+\infty, o, \oc)$ furnish the characterization of the
asymptotic infrared degrees of freedom of the outgoing and incoming free
fields respectively; $\z_A(+\infty, o, \oc)$ and $\z'_A(-\infty, o, \oc)$
describe the asymptotics of the Coulomb field of the outgoing and incoming
asymptotic currents respectively.

Further transformation of these long-range variables will prove useful.
From the homogeneity (\ref{zet1}) it follows $\dsp
o^{A'} \D_{A'} \z_A (s, o, \oc) + s\, \dot{\z}_A(s, o, \oc) =0$.
For the limit points $s=\pm\infty$ we obtain
$\dsp o^{A'} \D_{A'} \z_A (\pm\infty, o, \oc)=0$, or
$\D_{A'} \z_A (\pm\infty, o, \oc) \propto l_a$ (as at the same time
$o^A\D_{A'} \z_A (s, o, \oc)=0$ for all $s$). The same holds true for
$\z'_A(\pm\infty, o, \oc)$. The r.h. sides of the following equations
introduce new variables
\begin{align}
 \D_{A'}\z_A(+\infty, o, \oc) &= - l_a q(o, \oc)\, ,
\label{q}\\[1ex]
 \D_{A'}\z'_A(-\infty, o, \oc) &= - l_a q'(o, \oc)\, ,
\label{q'}\\[1ex]
 \D_{A'}\left(\z_A(-\infty, o, \oc) - \z_A(+\infty, o, \oc)\right)
 &= o_A\D_{A'}\z^{\mrm{out}}(-\infty, o, \oc) = - l_a \si(o, \oc)\, ,
\label{si}\\[1ex]
 \D_{A'}\left(\z'_A(+\infty, o, \oc) - \z'_A(-\infty, o, \oc)\right)
 &= o_A\D_{A'}\z'^{\mrm{in}}(+\infty, o, \oc) = - l_a \si'(o, \oc)\, .
\label{si'}
\end{align}
As a consequence of (\ref{zet3}) one has a constraint
\be
q+\si = q'+\si'\, .
\label{invar}
\ee
All of the new variables are spinor functions of the homogeneity
\be
f(\alpha o, \ac\oc) = (\alpha\ac)^{-2} f(o, \oc)\, ,
\label{hom}
\ee
where $f$ stands for any of $q$, $q'$, $\si$ or $\si'$. Moreover they satisfy
\begin{align}
 \frac{1}{2\pi}\int q(l)\, \dl &= \frac{1}{2\pi}\int q'(l)\, \dl = Q\, ,
 \label{mq}\\[1ex]
 \frac{1}{2\pi}\int \si(l)\, \dl &= \frac{1}{2\pi}\int \si'(l)\, \dl=0\, .
 \label{ms}
\end{align}
One calculates these means by
contracting (\ref{q})--(\ref{si'}) with a timelike, unit, future-pointing
vector, integrating
by parts and using (\ref{zet2}), e.g. for $q$ one has
\[
 -{1\over 2\pi}\int t^a\D_{A'}\z_A(+\infty, o, \oc)\,
{\dl\over t\s l} = -{1\over 4\pi}\int o^A\z_A(+\infty, o, \oc) \,
{\dl\over (t\s l)^2} = Q.
\]

Conversely, the conditions (\ref{hom}), (\ref{mq}) and (\ref{ms}) are the
only ones following from (\ref{q}) - (\ref{si'}) and the functions $q$, $q'$,
$\si$ and $\si'$ satisfying them determine the long-range variables
$\z_A(+\infty, o, \oc)$, $\z'_A(-\infty, o, \oc)$,
$\z^{\mrm{out}}(-\infty, o, \oc)$ and
$\z'^{\mrm{in}}(+\infty, o, \oc)$ uniquely. For the last two of them this
follows directly from (\ref{swh2}). To prove the statement for
the other two, we choose a vector $t^a$ and denote
$\io^A = (t\s l)^{-1} t^{AA'} o_{A'}$. Then
$\z_A(+\infty, o, \oc) = Q\io_A + \io^B\z_B(+\infty, o, \oc) o_A$
and (\ref{q}) is equivalent to
\[
 \D_{A'}\Big( \io^A\z_A(+\infty, o, \oc)\Big) =
-o_{A'} \Big( q(o, \oc) - {Q\over 2(t\s l)^2}\Big).
\]
The proof now ends
as for $\si$'s. Similarly for the primed quantities.

The vanishing of means (\ref{ms}) implies also that there exist homogeneous
functions of degree zero
\be
\F(\alpha o, \ac\oc) = \F(o, \oc)\, , ~~~\F'(\alpha o, \ac\oc) = \F'(o, \oc)\, ,
\label{Fi1}
\ee
such that
\be
\D_A\D_{A'}\F(o, \oc) = l_a\, \si(o, \oc)\, , ~~
\D_A\D_{A'}\F'(o, \oc) = l_a\, \si'(o, \oc)\, .
\label{Fi2}
\ee
These conditions determine $\F$'s up to additive constants. In view of
(\ref{si}) and (\ref{si'}) eqs (\ref{Fi2}) are equivalent to
\be
\D_A\F(o, \oc) = -o_A \z^{\mrm{out}}(-\infty, o, \oc)\, ,~~
\D_A\F'(o, \oc) = -o_A \z'^{\mrm{in}}(+\infty, o, \oc)\, .
\label{Fi3}
\ee
Physical meaning of the additive constant in $\F$ (and $\F'$) comes from the
fact, that gauges of potentials can be divided into equivalence classes,
each class remaining in one to one correspondence with a choice of this
constant. Classes differ by their infrared contributions to
symplectic form; see ref. \cite{her}.

With the use of the variables $q$,\ldots,$\si'$ further insight into the
meaning of (\ref{espas2}) is possible. We observe first that the result of
(\ref{espas2}), considered as a function of $y$, is a free electromagnetic
field in the region $y^2 < 0$, which is homogeneous of degree~$-2$. We
denote this field $\f^{\mrm{l.r.}}{}_{AB}(y)$ (l.r. standing for long-range).
Using the identity \mbox{$\dsp o_A = o^C\e_{CA} = {2\over y^2}
o^C y_{CD'}y^{D'}_A$} we can represent $\dsp \delta'(y\s l) o_A =
{2\over y^2} y^{D'}_A\D_{D'} \delta(y\s l)$. Setting this into (\ref{espas2})
and integrating by parts we get
\be
\f^{\mrm{l.r.}}{}_{AB}(y) = y^{C'}_{(A} K_{B)C'}(y)\, ,
\label{long}
\ee
where
\be
K_a(y) = {1\over 2\pi y^2} \n_a\int\sgn(y\s l)(q+\si)(o, \oc)\, \dl\, .
\label{Ka}
\ee
The tensor form of equation (\ref{long}) gives the electromagnetic field
$F_{ab}^{\mrm{l.r.}}$ corresponding to the spinor $\f^{\mrm{l.r.}}{}_{AB}(y)$
\be
F^{\mrm{l.r.}}{}_{ab} = F^{\mrm{l.r.E}}{}_{ab} +
F^{\mrm{l.r.M}}{}_{ab}\, ,
\label{lEM}
\ee
where the fields on the r.h. side are determined by
\begin{align}
 F^{\mrm{l.r.E}}{}_{ab}(y) &= \mrm{Re}K_a(y)y_b -
 \mrm{Re}K_b(y)y_a\, ,\label{lE}\\[1ex]
 {}^*F^{\mrm{l.r.M}}{}_{ab}(y) &= \mrm{Im}K_a(y)y_b - \mrm{Im}K_b(y)y_a\, .
\label{lM}
\end{align}
The above equations imply $^*F^{\mrm{l.r.E}}{}_{ab}(y)y^b =0$ and
$F^{\mrm{l.r.M}}{}_{ab}(y)y^b =0$. This can be interpreted as follows:
in any Minkowski frame the radial components of the magnetic part of the field
$F^{\mrm{l.r.E}}{}_{ab}$ and of the electric part of the field
$F^{\mrm{l.r.M}}{}_{ab}$ vanish. Equivalently formulated:
$F^{\mrm{l.r.E}}{}_{ab}$ has no magnetic multipole contributions and
$F^{\mrm{l.r.M}}{}_{ab}$ has no electric multipole contributions in any
Minkowski frame. Accordingly, $F^{\mrm{l.r.E}}{}_{ab}$ and
$F^{\mrm{l.r.M}}{}_{ab}$ will be called the electric- and the magnetic-type
long-range fields respectively (cf. \cite{sta89}).
The whole above discussion applies also to
the long-range part of free asymptotic fields ($q+\si$ replaced by $\si$ or
$\si'$ in (\ref{Ka})) and of the Coulomb fields of the asymptotic currents
($q+\si$ replaced by $q$ or $q'$). The real (imaginary) parts of
$q$,\ldots,$\si'$ describe the electric (magnetic) parts of the respective
fields.

We end this section by testing our assumptions on admissible currents
in two cases: a system of charged point particles and a free Dirac field.
A~system of $N$ point charges moving along trajectories
$z_{\mrm{i}}{}^a(\tau)$, $\mrm{i}=1,\ldots,N$, each
parametrized by its proper time, corresponds to an obviously spacelike
finitely extended current density
\be
J^a(x) = \sum_{\mrm{i}=1}^N Q_{\mrm{i}}\int\delta(x-z_{\mrm{i}}(\tau))\,
v_{\mrm{i}}{}^a(\tau)\, d\tau\, ,
\label{point1}
\ee
where $\dsp v^a(\tau)\equiv {dz^a\over d\tau}(\tau)$. The asymptotic
characteristic (\ref{elce}) is easily obtained
\be
c^A(s, o, \oc) = \sum_{\mrm{i}=1}^N \left[ {Q_{\mrm{i}}\over v_{\mrm{i}}
(\tau)\s l}
v_{\mrm{i}}{}^{AA'}(\tau)o_{A'} \right]
\Big|_{\textstyle \tau : s=z_{\mrm{i}}(\tau)\s l}\, .
\label{point2}
\ee
If the asymptotic behaviour of four-velocities satisfies
$\dsp \left|\frac{dv^a}{d\tau}(\tau)\right| <
\frac{\con}{|\tau|^{1+\e}}$ for $|\tau|\to \infty$, then $c^A(s, o, \oc)$ has
limits
\be
c^A(\pm\infty, o, \oc) = \sum_{\mrm{i}=1}^N {Q_{\mrm{i}}\over
v_{\mrm{i}}(\pm\infty)\s l}
v_{\mrm{i}}{}^{AA'}(\pm\infty)o_{A'}
\label{point3}
\ee
and the assumptions on the fall-off of
$|c^A(s, o, \oc) - c^A(\pm\infty, o, \oc)|$ for $s\to\pm\infty$ are satisfied.
The class of thus admitted asymptotic motions includes the typical
behaviour of the Coulomb scattering, where $\e=1$. From (\ref{point3})
we get
\begin{equation}\label{point4}
\dsp q(o, \oc) =
\sum_{\mrm{i}=1}^N {Q_{\mrm{i}}\over 2(v_{\mrm{i}}(+\infty)\s l)^2}\, ,\qquad
\dsp q'(o, \oc) =
\sum_{\mrm{i}=1}^N {Q_{\mrm{i}}\over 2(v_{\mrm{i}}(-\infty)\s l)^2}\, .
\end{equation}
If no magnetic monopoles are present, the $q$'s are real.

For the discussion of a free Dirac field we use its Fourier representation in
the following form
\be
\psi(x) = \left({m\over 2\pi}\right)^{3/2}\int
\left(e^{\textstyle -imx\s v} P_+f(v) - e^{\textstyle +imx\s v} P_-f(v)\right)
\, \m(v)\, ,
\label{Di1}
\ee
in the notation explained in the first paragraph of Section \ref{sdir}. If we
assume that $f(v)$ is an infinitely differentiable function of compact support,
then $\psi(x)$ is a~regular wave packet \cite{ree}, so it is a function of fast
decrease in all spacelike and lightlike directions \cite{ree}. The current
density has infinite extension in spacelike directions,
but its exponential fall-off is sufficient
for the extension of our results to this case.

The definition (\ref{elce}) can be rewritten in the Fourier-transformed form
\be
\tilde{c}^A(\w, o, \oc) = \hat{J}^{AA'}(\w l)o_{A'}\, ,
\label{cefour}
\ee
with the conventions introduced in (\ref{confour1}, \ref{confour2}).
From (\ref{Di1}) we get
\begin{align*}
 \hat{\psi}(mv) &= \left({2\pi\over m}\right)^{3/2}{2\over m}\delta(v^2-1)
 \left(\theta(v^0)P_+f(v) - \theta(-v^0)P_-f(-v)\right)\\[1ex]
 &\equiv \left({2\pi\over m}\right)^{3/2}{2\over m}\delta(v^2-1) h(v)\, ,
\label{Di2}
\end{align*}
so
\begin{align*}
 \hat{J}^a(mu) &= {em^4\over (2\pi)^3} \int\ba{\hat{\psi}}
 \lp m\lp r-{u\over 2}\rp \rp\,
 \g^a\,\hat{\psi}\lp m\lp r+{u\over 2}\rp\rp\, d^4r \\[1ex]
 &= {2e\over m}\int\delta(u\s r)\delta\left( r^2+{u^2\over 4}-1\right)
 \ba{h}\left( r-{u\over 2}\right)\g^a h\left( r+{u\over 2}\right)\, d^4r\, .
\end{align*}
Hence
\be
\hat{J}^a (\w l) = e\delta(\w)\int v^a \fb\g\s v f(v){\m(v)\over v\s l}\, .
\label{Di3}
\ee
The asymptotic characteristic $c_A(s, o, \oc)$ does not depend on $s$, as was
to be expected (a free charged field sends no radiation field)
\be
c^A(s, o, \oc) = e\int v^{AA'} \fb\g\s v f(v)\,{\m(v)\over v\s l}\, o_{A'}\, .
\label{Di4}
\ee
Not only the charge of the field
\be
Q = c^Ao_A = e\int\fb\g\s v f(v)\, \m(v)
\label{Di5}
\ee
but also the asymptotic variables
\be
q(o, \oc)=q'(o, \oc)= e\int\fb\g\s v f(v)\,{\m(v)\over 2(v\s l)^2}
\label{Di6}
\ee
are obviously real. The long-range electromagnetic field produced by
a free Dirac field is therefore of purely electric type
\be
F^{\mrm{l.r.}}{}_{ab} = e\int\fb\g\s v f(v)
{y_av_b - y_bv_a\over \left((v\s y)^2-y^2\right)^{3/2}}\,\m(v)\, .
\label{Di7}
\ee

The absence of magnetic-type long-range fields is a typical feature of the
scattering processes involving no magnetic monopoles. To produce a~long-range
magnetic-type field without the use of magnetic monopoles one would need
an asymptotic current of infinitely increasing magnetic multipoles,
a~magnetic dipole linearly growing with time giving the simplest possibility.
As the infrared singular free fields are typically produced as radiation
fields of some scattering processes they also yield the long-range fields
of electric type only.

\setcounter{equation}{0}
\setcounter{pr}{0}
\section{Energy-momentum and angular momentum tensor of the asymptotic
electromagnetic field}
\label{srad}

Consider now a closed dynamical system, part of which constitute the
(generalized) Maxwell equations (\ref{max1}) with the current $J_a$
satisfying the assumptions of the previous section. Finite spacelike
extension of $J_a$, which we assume for simplicity, could be replaced
by some fast decrease condition, more appropriate in the case where
$J_a$ is due to some charged massive field. Suppose further that the
system is equipped with a locally conserved, symmetric energy-momentum
tensor $T_{ab}$,
which outside the electromagnetic sources reduces to the usual symmetric
electromagnetic tensor $T^{\mrm{elm}}{}_{ab} = -{1\over 4\pi}
\big(F_{ac}F_b{}^c - {1\over 4}g_{ab}F_{cd}F^{cd}\big)$, and the amount of
energy-momentum and angular momentum passing through a hypersurface $\Sc$
is given as usual respectively by
\begin{align}
 P_a[\Sc] &= \int_\Sc T_{ac}(x)\, d\si^c(x)\, ,\label{em}\\[1ex]
 M_{ab}[\Sc] &= \int_\Sc \lp x_aT_{bc}(x) - x_bT_{ac}(x)\rp\, d\si^c(x)\, .
\label{amt}
\end{align}
In the context of electromagnetic fields it proves convenient to use the spinor
version of equation (\ref{amt}). If the symmetric angular momentum spinor
$\mu_{AB}$ is defined by
\be
M_{ab} = \mu_{AB}\e_{A'B'} + \ba{\mu}_{A'B'}\e_{AB}\, ,
\label{amts}
\ee
then (\ref{amt}) is equivalent to
\be
\mu_{AB}[\Sc] = \int_\Sc \mu_{ABc}(x)\, d\si^c(x)\, ,
\label{ams1}
\ee
where
\be
\mu_{ABc}(x) = x_{D'(A}T^{D'}_{B)c}(x)\, .
\label{ams2}
\ee

We want to consider now the total energy-momentum and angular momentum of
the system and express these quantities in terms of asymptotic fields.
The usual straightforward expressions, obtained by setting for $\Sc$ in
(\ref{em}) and (\ref{amt}) a Cauchy surface $\Sigma$, are not appropriate for
our purpose for two reasons.

(i) The asymptotic electromagnetic fields discussed in the preceding section
are defined in the null asymptotic region, whereas one should expect the
asymptotic massive fields to be defined in timelike asymptotic regions
(we shall return to this question in Section \ref{sdir}). This physical
picture suggests that separation of the contributions to the conserved
quantities could be possible. We shall see that this is almost true, the
reservation representing a physically important term in the total angular
momentum involving long-range Coulomb and infrared degrees of freedom.
For the demonstration of this separation the Cauchy surface
integration is not well suited, as this surface contains the whole
information on the system, even if it is pushed to infinite past or future.

(ii) In the case of angular momentum an even more serious obstacle arises:
the integrand of (\ref{amt}) is not absolutely integrable for a Cauchy
surface, so, strictly speaking, the integral does not exist. This is
easily seen from the spacelike asymptotic behavior of electromagnetic field,
discussed in the preceding section. On a hypersurface $x^0=\con$ the field is
$O(|\vec{x}|^{-2})$, so the integrand is $O(|\vec{x}|^{-3})$, while
the measure is $d^3x$.

With our purpose in mind we consider first the energy-momentum and the angular
momentum radiated with the electromagnetic field into future null directions.
Let $a$ be a point vector of arbitrary point in Minkowski space and $t$
a timelike, unit, future-pointing vector.
Choose the line $a+\tau\, t$, $\tau\in\mR$, as the
time-axis of the origin of $3$-space orthogonal to $t$. Consider the timelike
tube given by $x= a+\tau\, t+R\, l$, $R=\con$, $\tau\in\mR$, $l$
going over the set of all future null vectors in the $t$-gauge $t\s l=1$.
The energy-momentum and the angular momentum passing
through a bounded portion of
this tube are given by
\be
\int_B\rho_c(a+\tau\, t+R\, l)(t^c-l^c)\, d\tau\,R^2d\W_t(l)\, ,
\label{tube}
\ee
where $\rho_c = T_{ac}$ for energy-momentum and $\rho_c=\mu_{ABc}$ for
angular momentum spinor, and integration extends over a bounded interval
of retarded time $T$ and a solid angle $\Theta$ of $l$-directions. The limit
of the above expressions when $R\to\infty$,
if it exists, gives the respective quantities
radiated into the solid angle $\Theta$ over the time-span $T$. More general
bounded measurable sets of integration $B$ are possible. For sufficiently
large $R$ we move into the region where sources vanish and $T_{ab}=
T^{\mrm{elm}}{}_{ab}$. In the spinor language the latter takes the form
\be
T^{\mrm{elm}}{}_{ab}(x) = {1\over 2\pi} \fc_{A'B'}(x)\ex\, ,
\label{emt}
\ee
which yields
\be
\mu^{\mrm{elm}}{}_{ABc}(x) =
-{1\over 2\pi}\fac_{C'(A}(x)\f_{B)C}(x)\, ,
\label{ams3}
\ee
with $\fa_{AA'}$ defined in (\ref{fa}).
We see now, that for the energy-momentum and the angular momentum radiated
over finite time intervals to be well defined, both $\f_{AB}(x)$ and
$\fa_{AA'}(x)$ must have the null asymptotics of the assumed type.
Then
\begin{align}
 \lim_{R\to\infty} R^2\, T^{\mrm{elm}}{}_{ac}(a+\tau\, t+R\, l)
 &= {1\over 2\pi}\ba{\dot{\z}}_{A'}\dot{\z}_A(a\s l +\tau, o, \oc)\, l_c\, ,
 \label{rad1}\\[1ex]
 \lim_{R\to\infty} R^2 \,\mu^{\mrm{elm}}{}_{ABc}(a+\tau\, t+R\, l)
 &= -{1\over 2\pi}\ba{\nu}_{(A}\dot{\z}_{B)}(a\s l +\tau, o, \oc)\, l_c\, ,
\label{rad2}
\end{align}
where (\ref{elas3}) and (\ref{elas4}) have been used (in $t$-gauge). This
justifies our assumptions on the asymptotic behaviour of electromagnetic
field (Section \ref{selm}). Using the trick described in Section \ref{selm}
after eq.(\ref{espas1}) and a bound analogous to that following
eq.(\ref{freeas2}) one easily shows that for large $R$ the quantities under
the limits on the l.h. sides of (\ref{rad1}) and (\ref{rad2}) are
bounded by $\con R^{-1}$ on any bounded set $B$. The limits may be thus
performed under the integral sign in (\ref{tube}), which yields the
radiated quantities
\begin{align}
 P^{\mrm{out-n}}{}_a[B] &= {1\over 2\pi}\int_B
 \ba{\dot{\z}}_{A'}\dot{\z}_A(a\s l +\tau, o, \oc)\, d\tau\, d\W_t(l)\, ,
  \label{rad3}\\[1ex]
 \mu^{\mrm{out-n}}{}_{AB}[B] &= -{1\over 2\pi}\int_B
 \ba{\nu}_{(A}\dot{\z}_{B)}(a\s l +\tau, o, \oc)\, d\tau\, d\W_t(l)\, ,
 \label{rad4}
\end{align}
n standing for null.
The fall-off of asymptotics is sufficient for the integrals in (\ref{rad3})
and (\ref{rad4}) to be absolutely integrable over any measurable set $B$,
not necessarily bounded. Thus extension of the range of integration $B$ to
all times and full solid angle is possible. The result is
\begin{align}
 P^{\mrm{out-n}}{}_a &= {1\over 2\pi}\int
 \ba{\dot{\z}}_{A'}\dot{\z}_A(s, o, \oc)\, ds\, \dl\, ,
 \label{rad5}\\[1ex]
 \mu^{\mrm{out-n}}{}_{AB} &= -{1\over 2\pi}\int
 \ba{\nu}_{(A}\dot{\z}_{B)}(s, o, \oc)\, ds\, \dl\, ,
  \label{rad6}
\end{align}
where no assumption on the gauge of spinors is needed any more and any
reference to the time-axis (the point vector $a$ and the vector $t$)
has been lost. Thus the total radiated quantities are unambiguously defined.

In the next step we turn to the energy-momentum and angular momentum going
out with the massive part of the system in timelike directions. We choose
again the time-axis $a+\tau\, t$, fix a time parameter $\tau=\tau_1$ and
consider the future lightcone $\Co^{\mrm{fut}}(\tau_1)$ with the
vertex in $a+\tau_1\, t$. The amount of energy-momentum and angular
momentum passing through that cone is that contained in the system when
the quantities radiated prior to $\tau=\tau_1$ are disregarded.
For this interpretation to make sense the appropriate integrals over the cone
should be absolutely convergent. That this is the case here,
can be seen with the use of the trick
described in the preceding section, the estimates for a free field discussed
in \ref{apb} and properties of the Coulomb field (\ref{Coul}). In this way
we obtain quantities $P_a[\Co^{\mrm{fut}}(\tau_1)]$ and
$\mu_{AB}[\Co^{\mrm{fut}}(\tau_1)]$. (We have tacitly assumed that there
are no nonintegrable singularities in the region of nonvanishing sources.)
Due to the assumed local
conservation of the energy-momentum tensor the difference
of quantities calculated on two different cones $\Co^{\mrm{fut}}(\tau_2)$
and $\Co^{\mrm{fut}}(\tau_1)$ is given by the integrals (\ref{rad3})
and (\ref{rad4}), with $B=|\tau_2-\tau_1|\times\{\mrm{full~solid~angle}\}$.
The convergence of integrals (\ref{rad5}) and (\ref{rad6}) implies now
the existence of the limits
\begin{align}
 P^{\mrm{out-t}}{}_a&\equiv\lim_{\tau\to\infty}
 P_a[\Co^{\mrm{fut}}(\tau)]\, ,
 \label{outm1}\\[1ex]
 \mu^{\mrm{out-t}}{}_{AB}&\equiv\lim_{\tau\to\infty}
 \mu_{AB}[\Co^{\mrm{fut}}(\tau)]\, ,
 \label{outm2}
\end{align}
t standing for timelike. By a similar argument, with the use of two different
time-axis, the quantities thus obtained are independent of the choice of
the axis. Moreover, instead of lightcones any timelike hypersurfaces tending
to them asymptotically can be used. In the case of free electromagnetic field
one should expect that all energy-momentum and angular momentum are
radiated into null directions. This is indeed the case, as shown in \ref{apb},
i.e. we have
\be
 P^{\mrm{out-t}}{}_a(\mrm{free})=0,\qquad
 \mu^{\mrm{out-t}}{}_{AB}(\mrm{free})=0\, .
\label{freefut}
\ee
In this way we are led to unambiguous interpretation of (\ref{outm1}) and
(\ref{outm2}) as quantities going out with the massive part of the system.
For an explicit representation in terms of dynamical asymptotic variables
one needs more detailed knowledge of the system. We discuss this question
in the following sections.

The preceding discussion strongly suggests the identification of the total
energy-momentum and angular momentum of the system by
\begin{align}
 P^{\mrm{tot}}{}_a &= P^{\mrm{out-n}}{}_a + P^{\mrm{out-t}}{}_a\, ,
 \label{tot1}\\[1ex]
 \mu^{\mrm{tot}}{}_{AB} &= \mu^{\mrm{out-n}}{}_{AB}
 + \mu^{\mrm{out-t}}{}_{AB}\, .
\label{tot2}
\end{align}
Two points in this connection have to be clarified.\\
(i) The connection of (\ref{tot1}) and (\ref{tot2}) to the quantities
obtained by the Cauchy surface integration, if it can be performed,
should be understood. \\
(ii) The picture lying at the base of our discussion can be reflected
in time, with the subsequent change of orientation of hypersurfaces.
The radiated quantities are then replaced by the respective quantities
incoming from the past null directions, given by
\begin{align}
 P^{\mrm{in-n}}{}_a &= {1\over 2\pi}\int
 \ba{\dot{\z'}}_{A'}\dot{\z'}_A(s, o, \oc)\, ds\, \dl\, ,
 \label{innu1}\\[1ex]
 \mu^{\mrm{in-n}}{}_{AB} &= -{1\over 2\pi}\int
 \ba{\nu'}_{(A}\dot{\z'}_{B)}(s, o, \oc)\, ds\, \dl\, .
\label{innu2}
\end{align}
Similarly, the past timelike limits of integrals over past lightcones,
$P^{\mrm{in-t}}{}_a$ and $\mu^{\mrm{in-t}}{}_{AB}$, replace
(\ref{outm1}) and (\ref{outm2}) respectively, and again
\be
P^{\mrm{in-t}}{}_a(\mrm{free})=0,
~~\mu^{\mrm{in-t}}{}_{AB}(\mrm{free})=0\, .
\label{freepast}
\ee
For the consistence of physical
interpretation formulae (\ref{tot1}) and (\ref{tot2}) should yield the
same results with in-quantities on the r.h. sides.

To clarify the above raised points consider the situation depicted in Fig.1.
We choose an arbitrary spacelike hyperplane $\Sigma$ and a time-axis with
the unit vector $t$ orthogonal to $\Sigma$,
crossing the plane at the point $a$.
$\Co^{\mrm{fut}}(\tau)$, $\tau>0$, is the future lightcone with the vertex in
$a+\tau\, t$; $\Co'^{\mrm{fut}}(-r)$, $r>0$, is the unbounded portion of
the future lightcone with the vertex in $a-rt$, cut off by the plane
$\Sigma$. $\Co^{\mrm{past}}(-\tau)$ and $\Co'^{\mrm{past}}(r)$ are
obtained from the former two cones by reflection with respect to $\Sigma$.
$\Sigma(r)$ is the portion of $\Sigma$ (a ball) closing the cut of
$\Co'^{\mrm{fut}}(-r)$ and $\Co'^{\mrm{past}}(r)$. We consider
conservation of energy-momentum and angular momentum for two infinite regions:
the first contained between $\Co^{\mrm{fut}}(\tau)$,
$\Co'^{\mrm{fut}}(-r)$ and $\Sigma(r)$, the second contained between
$\Co^{\mrm{past}}(-\tau)$, $\Co'^{\mrm{past}}(r)$
and $\Sigma(r)$. Taking the limits $\tau\to\infty$ and $r\to\infty$ we arrive
at
% ##################
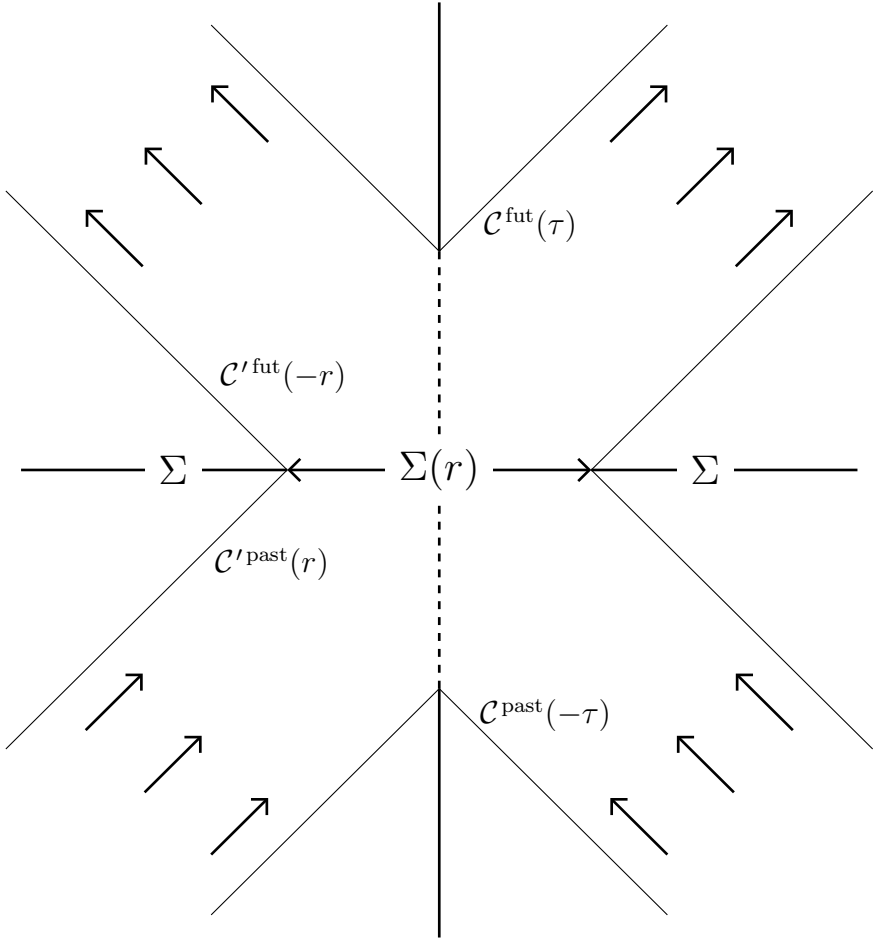
\begin{figure}

  \label{fig:-01}

  \centering

  \begin{tikzpicture}

    % Horizontal axis
    \draw[axis line] (-5.5,0) -- (-2,0);

    \draw[axis arrow] (0,0) -- (-2,0);

    \draw[axis arrow] (0,0) -- (2,0);

    \draw[axis line] (2,0) -- (5.5,0);

    % Vertical axis
    \draw[axis line] (0,-6.2) -- (0,-2.9);

    \draw[axis line,dashed] (0,-2.9) -- (0,2.9);

    \draw[axis line] (0,2.9) -- (0,6.2);

    % Hypersurface labels
    \node[scale=1.3,fill=white] at (0,0) {$\Sigma( r )$};

    \node[scale=1.3,fill=white] at (-3.5,0) {$\Sigma$};

    \node[scale=1.3,fill=white] at (3.5,0) {$\Sigma$};

    % Light cone boundaries

    % Left boundary
    \draw (-5.7,3.7) -- (-2,0) -- (-5.7,-3.7);

    % Upper boundary
    \draw (-3,5.9) -- (0,2.9) -- (3,5.9);

    % Right boundary
    \draw (5.7,3.7) -- (2,0) -- (5.7,-3.7);

    % Down boundary
    \draw (-3,-5.9) -- (0,-2.9) -- (3,-5.9);

    % Vectors shorwing propagation of the field

    % Left upper region
    \draw[pointing arrow] (-3.9,2.7) -- (-4.65,3.45);

    \draw[pointing arrow] (-3.125,3.525) -- (-3.875,4.275);

    \draw[pointing arrow] (-2.25,4.35) -- (-3,5.1);

    \node at (-2.05,1.25)
    {$\mathcal{C}\hspace{0.07em}'^{ \hspace{0.14em} \textrm{fut} }( -r )$};

    % Right upper region
    \draw[pointing arrow] (3.9,2.7) -- (4.65,3.45);

    \draw[pointing arrow] (3.125,3.525) -- (3.875,4.275);

    \draw[pointing arrow] (2.25,4.35) -- (3,5.1);

    \node at (1.2,3.25) {$\mathcal{C}^{ \hspace{0.09em} \textrm{fut} }( \tau )$};

    % Left down region
    \draw[pointing arrow] (-4.65,-3.45) -- (-3.9,-2.7);

    \draw[pointing arrow] (-3.875,-4.275) -- (-3.125,-3.525);

    \draw[pointing arrow] (-3,-5.1) -- (-2.25,-4.35);

    \node at (-2.2,-1.25)
    {$\mathcal{C}\hspace{0.07em}'^{ \hspace{0.12em} \textrm{past} }( r )$};

    % Right down region
    \draw[pointing arrow] (4.65,-3.45) -- (3.9,-2.7);

    \draw[pointing arrow] (3.875,-4.275) -- (3.125,-3.525);

    \draw[pointing arrow] (3,-5.1) -- (2.25,-4.35);

    \node at (1.4,-3.25) {$\mathcal{C}^{ \hspace{0.09em} \textrm{past} }( -\tau )$};

  \end{tikzpicture}

  \caption{The choice of hypersurfaces for the derivation of Eqs (\ref{out1})
and (\ref{in1}) (two spacelike dimensions are suppressed).}

\end{figure}
% ##################
\begin{align}
 G^{\mrm{out-n}} + G^{\mrm{out-t}} &=\lim_{r\to\infty}
 \lp G\/[\Sigma(r)] + G\/[\Co'^{\mrm{fut}}(-r)]\rp\, ,\label{out1}\\[1ex]
 G^{\mrm{in-n}} + G^{\mrm{in-t}} &=\lim_{r\to\infty}
 \lp G\/[\Sigma(r)] + G\/[\Co'^{\mrm{past}}(r)]\rp\, ,\label{in1}
\end{align}
where $G$ stands for $P_a$ or $\mu_{AB}$. Conventions of hypersurface
orientations are such, that positive direction of crossing the surface is
from past to future. The limits on the r.h. sides of equations
(\ref{out1}) and (\ref{in1}) exist, since the l.h. sides exist.
We shall show that those limits exist indeed for each term on the r.h.
sides separately and the following explicit formulae hold true
\begin{align}
 &\ \lim_{r\to\infty} P_a[\Co'^{\mrm{fut}}(-r)]
 =\lim_{r\to\infty} P_a[\Co'^{\mrm{past}}(r)] =0 \, ,
\label{reste}\\[1ex]
 &\begin{aligned}
 \lim_{r\to\infty} \mu_{AB}[\Co'^{\mrm{fut}}(-r)]
 &= -\lim_{r\to\infty} \mu_{AB}[\Co'^{\mrm{past}}(r)] \\
 &= {1\over 4\pi}
 \int\ba{\nu}_{(A}\z_{B)}(-\infty, o, \oc)\, \dl\, .
\end{aligned}
\label{resta}
\end{align}
It is easy to see, that if these formulae are shown for the special choice
of the point $a=0$, then they remain true for other time-axes, so we
consider this special case.
We observe first that for sufficiently large $r$ both
$\Co'^{\mrm{fut}}(-r)$ and $\Co'^{\mrm{past}}(r)$ enter the region where
we can use the trick of the preceding section. Contributions of the
quadratic free field terms, Coulomb terms and mixed terms can be considered
separately. For the free field one easily shows with the use of
Lemma \ref{lemsi} that
\begin{align}
 \lim_{r\to\infty} P_a[\Sigma(r)]\Big|_{\mrm{free}}
 &=P^{\mrm{out-n}}{}_a\Big|_{\mrm{free}} =
 P^{\mrm{in-n}}{}_a\Big|_{\mrm{free}}\, , \label{emcf}\\[1ex]
 \lim_{r\to\infty} \mu_{AB}[\Sigma(r)]\Big|_{\mrm{free}}
 &= \frac{1}{2}\lp \mu^{\mrm{out-n}}{}_{AB} + \mu^{\mrm{in-n}}{}_{AB}\rp
 \Big|_{\mrm{free}} \, . \label{amcf}
\end{align}
Setting this into (\ref{out1}) and (\ref{in1}), taking into account
(\ref{freefut}) and (\ref{freepast}) and using (\ref{rad5}), (\ref{rad6}),
(\ref{innu1}) and (\ref{innu2}) for free fields one arrives at
(\ref{reste}) and (\ref{resta}). For mixed term one shows by a direct
calculation demonstrated in \ref{apb} that (\ref{reste}) holds and
the r.h. side of (\ref{resta}) takes the required form
\[
{1\over 4\pi} \int\lp\ba{\nu^Q}_{(A}\z^{\mrm{free}}{}_{B)}
(-\infty, o, \oc) + \ba{\nu^{\mrm{free}}}_{(A}\z^Q{}_{B)}
(-\infty, o, \oc)\rp\,\dl.
\]
For the Coulomb terms all the terms in (\ref{reste}) and (\ref{resta})
vanish. This ends the proof of (\ref{reste}) and (\ref{resta}).

We return to the physical interpretation. Consistent identification of the
total energy-momentum by
\be
P^{\mrm{in}}{}_a = P^{\mrm{out}}{}_a = P_a[\Sigma]
\label{emtot}
\ee
is always correct due to (\ref{reste}); on the r.h. side the proper integral
replaces the limit, as the integrand is absolutely integrable. For angular
momentum we have to impose a (Poincar\'e covariant) condition on the
long-range variables
\be
{1\over 4\pi} \int\ba{\nu}_{(A}\z_{B)}(-\infty, o, \oc) =0
\label{exist}
\ee
to be able to conclude
\be
\mu^{\mrm{in}}{}_{AB} = \mu^{\mrm{out}}{}_{AB} =
\lim_{r\to\infty}\mu_{AB}[\Sigma(r)]\, .
\label{amtot}
\ee
Crucial for the interpretation is the first equality. The Cauchy integral is
not absolutely integrable, but the second equality gives its finite
regularization, which, however, has no independent direct physical
justification. (The mechanism of this regularization is the antisymmetry
of the leading asymptotic term of the integrand with respect to reflection
of 3-space $\Sigma$ -- this is easily seen from (\ref{long}) and (\ref{Ka}).)
In the absence of the condition (\ref{exist}) no well founded identification
of angular momentum seems to be possible -- angular momentum leaks out into
the spacelike infinity.
Condition (\ref{exist}) imposes constraints on the long-range field
(\ref{long}). We recall that $o_A\nu_{A'}(-\infty, o, \oc)=
\D_{A'}\z_A(-\infty, o, \oc)= -(q+\si)\, l_a$, and decompose
$\z_A(-\infty, o, \oc)$ and $\nu_{A'}(-\infty, o, \oc)$ into their
electric- and magnetic-type parts:
\begin{align*}
 o_A\nu^{\mrm{E}}{}_{A'}(-\infty, o, \oc) &=
 \D_{A'}\z^{\mrm{E}}{}_A(-\infty, o, \oc)= - \mrm{Re}(q+\si)l_a\,,\\[1ex]
 \text{and}\qquad o_A\nu^{\mrm{M}}{}_{A'}(-\infty, o, \oc) &=
 \D_{A'}\z^{\mrm{M}}{}_A(-\infty, o, \oc)= - i\mrm{Im}(q+\si)l_a.
\end{align*}
Then also
\[
 o_{A'}\ba{\nu^{\mrm{E}}}_A(-\infty, o, \oc) =
 \D_{A'}\z^{\mrm{E}}{}_A(-\infty, o, \oc)\quad \text{but}\quad
 o_{A'}\ba{\nu^{\mrm{M}}}_A(-\infty, o, \oc) =
 -\D_{A'}\z^{\mrm{M}}{}_A(-\infty, o, \oc).
\]
Choose an arbitrary timelike, unit, future-pointing vector $t$
and contract the two last equations
with the associated spinor $\io^{A'}$. Setting the result into (\ref{exist})
and integrating by parts we have
$$
\int\io^{C'}\D_{C'}\z^{\mrm{E}}{}_{(A}
\z^{\mrm{M}}{}_{B)}(-\infty, o, \oc)\, \dl =0\, .
$$
The electric and the magnetic parts are independent,
and we know that electric-type
fields may be present. Therefore we demand that magnetic-type fields do not
occur in the theory. This condition, as shown in the preceding section,
is satisfied in the known situations of scattering phenomena. On the other
hand, for the field of freely moving charged (possibly both electrically
and magnetically) particles the l.h. side of (\ref{exist}) takes the form
$\dsp \sum_{\mrm{i}<\mrm{k}} \lp Q_{\mrm{i}}\ba{Q_{\mrm{k}}}
- Q_{\mrm{k}}\ba{Q_{\mrm{i}}}\rp h(v_{\mrm{i}}\s v_{\mrm{k}})\,
v_{\mrm{i}}{}_{C'(A}v_{\mrm{k}}{}^{C'}_{B)}$,
where $h$ is a real function (use (\ref{point3})
as asymptotic). This vanishes identically only if the ratios of the magnetic to
the electric charge are equal for all particles. This excludes presence of
magnetic charges, as pure electric charges have to be admitted. Finally,
we extend the condition of no magnetic part also to free fields, and
assume accordingly from now on
\be
\ba{q} = q\, ,~~\ba{q'} = q'\, ,~~\ba{\si} = \si\, ,~~
\ba{\si'} = \si'\, ,~~\ba{\F} = \F \, ,~~\ba{\F'} = \F'\, .
\label{real}
\ee

With the knowledge gained on the long-range degrees of freedom we turn again
to the quantities radiated into or coming from null directions. With the usual
definitions of outgoing and incoming free fields ($\f^{\mrm{out}}{}_{AB} =
\el - \f^{\mrm{ret}}{}_{AB}$, $\f^{\mrm{in}}{}_{AB} = \el
- \f^{\mrm{adv}}{}_{AB}$)
the asymptotics split according to (\ref{zet4}), (\ref{zet5}). Using these
splittings in (\ref{rad5}), (\ref{rad6}), (\ref{innu1}) and (\ref{innu2})
we obtain
\begin{eqnarray}
&&P^{\mrm{out-n}}{}_a = {1\over 2\pi}\int l_a\, \ba{\dot{\z}^{\mrm{out}}}
\dot{\z}^{\mrm{out}}(s, o, \oc)\, ds\, \dl\, ,
\label{split1}\\
&&\begin{array}{@{}l}
\dsp \mu^{\mrm{out-n}}{}_{AB} = -{1\over 2\pi}\int o_{(A}\D_{B)}
\ba{\z^{\mrm{out}}}\dot{\z}^{\mrm{out}}(s, o, \oc)\, ds\, \dl \\
\dsp \hspace*{50mm} + {1\over 2\pi}\int q o_{(A}\D_{B)}
\F(o, \oc)\, \dl\, ,
\end{array}
\label{split2}\\
&&P^{\mrm{in-n}}{}_a = {1\over 2\pi}\int l_a\,\ba{\dot{\z'}^{\mrm{in}}}
\dot{\z'}^{\mrm{in}}(s, o, \oc)\, ds\, \dl\, ,
\label{split3}\\
&&\begin{array}{@{}l}
\dsp \mu^{\mrm{in-n}}{}_{AB} = -{1\over 2\pi}\int o_{(A}\D_{B)}
\ba{\z'^{\mrm{in}}}\dot{\z'}^{\mrm{in}}(s, o, \oc)\, ds\, \dl \\
\dsp \hspace*{50mm} - {1\over 2\pi}\int q' o_{(A}\D_{B)}
\F'(o, \oc)\, \dl\, ,
\end{array}
\label{split4}
\end{eqnarray}
The first terms on the r.h. sides of all above equations are the pure free
field quantities. However, there are additional angular momentum contributions
due to the long-range tail of the electromagnetic field. These terms mix free
electromagnetic field characteristics $\si$, $\si'$ with the Coulomb
characteristics of asymptotic currents $q$, $q'$. To illustrate a possible
observational consequence of these additional terms we present a very
heuristic argument based on a guess on possible asymptotic states of
the theory.

Suppose that the timelike in-asymptotic of a scattering process is
characterized by a single massive spinless particle, carrying charge $Q$,
energy-momentum $P^{\mrm{in-t}}{}_a = m\, v_a$ and angular momentum\\
\mbox{$M^{\mrm{in-t}}{}_{ab} = m\, \lp y^{\mrm{in}}{}_a v_b -
y^{\mrm{in}}{}_b v_a \rp$}, where $y^{\mrm{in}}{}_a$ is a point-vector of
any point on the trajectory of the particle. Suppose further that the dynamics
of the theory supports an "adiabatic limit" characterized by the following
statements on the scattering states. We assume that the electromagnetic
in-field is infinitely low-energetic, that is
\[
 \int t\s l\big|\dot{\z}'^{\mrm{in}}(s, o, \oc)\big|^2 ds\, \dl
\to 0,
\]
with the infrared characteristic $\si'$, however, remaining finite.
We guess that there is then no particle production, no energy transfer,
so $P^{\mrm{out-t}}{}_a = m\, v_a$, and no radiation field, so the free
electromagnetic out-field is identical with the in-field. In consequence
there is $\si = \si'$ and $\dsp q=q' = \frac{Q}{2(v\s l)^2}$.
The scattering
process, nevertheless, is not completely trivial, since according to
(\ref{amtot}), (\ref{split2}) and (\ref{split4}) there is an angular
momentum change of the particle
\begin{eqnarray*}
&&\mu^{\mrm{out-t}}{}_{AB} - \mu^{\mrm{in-t}}{}_{AB} =
-{1\over \pi}\int q o_{(A}\D_{B)}\F(o, \oc)\, \dl \\
&&\hspace*{1cm}= {1\over \pi}\int \F o_{(A}\D_{B)}q(o, \oc)\, \dl =
{Q\over \pi}\int o_{(A} v_{B)}^{C'} o_{C'}\F(o, \oc)\,
\frac{\dl}{(v\s l)^3}\, ,
\end{eqnarray*}
or in the tensor language
$$
M^{\mrm{out-t}}{}_{ab} = m\lp y^{\mrm{out}}{}_a v_b -
y^{\mrm{out}}{}_b v_a \rp \, ,
$$
where
\be
y^{\mrm{out}}{}_a = y^{\mrm{in}}{}_a + \frac{Q}{\pi m} \int l_a
\F(o, \oc)\, \frac{\dl}{(v\s l)^3}\, ,
\label{shift}
\ee
determined up to a multiple of $v_a$ (this freedom corresponding to
$\F\to\F +\con$).
The effect of the scattering is thus an adiabatic translation of the
trajectory of the particle. This kind of effect will not show up in the
usual scattering cross-section measurements. To get an idea about the size
of the effect let us assume that the infrared characteristic of the free
incoming field is that of a field radiated by a charge $Q_0$ if it
changes its 4-velocity from $u_1$ to $u_2$. In that case
\[
 \z'^{\mrm{free}}{}_A(+\infty, o, \oc) = Q_0 \lp
\frac{u_2{}_{AC'}}{u_2\s l} - \frac{u_1{}_{AC'}}{u_1\s l}\rp o^{C'},
\]
so that up to an arbitrary constant $\dsp \F(o, \oc) = Q_0
\ln\frac{u_1\s l}{u_2\s l}$. One calculates then
\[
 \Delta y\equiv y^{\mrm{out}} - y^{\mrm{in}} =
\frac{QQ_0}{m}\lp f(\mrm{arcosh} u_2\s v)u_2 -
f(\mrm{arcosh} u_1\s v)u_1\rp,
\]
where $\dsp f(\beta) =
\frac{\sinh\beta\cosh\beta -\beta}{\sinh^3\beta}$. Assume for simplicity the following
experimental arrangement: in the laboratory system the energy of the
particle producing the incoming free field
remains constant and its $3$-velocity
is adiabatically reflected $\vec{u}_2 = - \vec{u}_1$; the test particle is
almost at rest in laboratory, so we neglect $|\vec{v}|$; both particles
are taken to be "electrons" (but with spin neglected).
Then $\Delta y_a$ is a translation in the $3$-space
of the laboratory and \mbox{$\dsp \|\Delta y_a\|=  r_{\mrm{cl}}
\frac{\sinh \xi\cosh\xi - \xi}{\sinh^2\xi}$}, where $\cosh\xi=u_2\s v=u_1\s v$
and $\dsp r_{\mrm{cl}}=\frac{e^2}{m}
\approx 2.8\times 10^{-13}\mrm{cm}$ is
"the classical radius of electron". The maximum value of the displacement
is of the order of $r_{\mrm{cl}}$, which is not too impressive. However,
the effect cumulates by multiple sending of identical incoming fields.

The first calculation of an observable effect produced by a free zero
frequency field is due to Staruszkiewicz \cite{sta}. He calculates, in the
quasiclassical approximation, the change of the phase of the wave function
of a particle in an external electromagnetic field. The plane wave
$e^{\textstyle -imv\s x}$ ($v$ is a four-velocity) undergoes in the complete
process of scattering by the field (\ref{freepot1}) the change of phase
$$
\delta(v) = -\frac{Q}{2\pi}\int\frac{v\s V(-\infty, l)}{v\s l}\, \dl
$$
(this is eq. (6) of \cite{sta} in our notation). Every gauge of $V_a(-\infty, l)$
can be represented by $V_a(-\infty, l) = \D_A h_{A'}(o, \oc) +
\D_{A'} \bar{h}_A(o, \oc)$ with some choice of the function $h_{A'}$
satisfying $h_{A'}(\alpha o, \ac\oc) = \ac^{-1} h_{A'}(o, \oc)$,
$h_{A'}(o, \oc) o^{A'} = \F(o, \oc)$ (cf. \cite{her}). Using this representation
and integrating by parts we obtain
\be
\delta(v) = -\frac{Q}{2\pi}\int\frac{\F(l)}{(v\s l)^2}\, \dl\, .
\label{fasta}
\ee
If a wave packet is formed, then this phase induces exactly the shift of the
trajectory given by (\ref{shift}).

\setcounter{equation}{0}
\setcounter{pr}{0}
\section{Dirac equation in the forward lightcone}
\label{sdir}

Our aim in this section is to reformulate the Dirac equation for an electron in
electromagnetic field as an evolution equation on the set of hyperboloids
$x^2 = \la^2$, $x^0 >0$, with $\la$ taking the role of evolution
parameter. We show
next that, under certain assumptions, scattering states exist. The class of
admitted potentials includes those Coulomb-like as well, if an appropriate
gauge transformation is performed.

To fix our notation we rewrite some standard facts about free Dirac field.
The Cauchy problem for the free Dirac equation
$$
( i\g\s\n - m) \p (x) = 0
$$
is solved by
\be
\p(x) = \int_\Sc S(x-y)\g^a\p(y)\,d\sigma_a(y)\, ,
\label{cauchy}
\ee
where $\Sc$ is a spacelike hypersurface and the Fourier representation of
$S(x)$ can be written as
$$
S(x) =\lp{m\over 2\pi}\rp^3 \int e^{\textstyle {-im\,x \s v\,\g\s v}}
\g\s v\, \m(v)\, ,
$$
where $\m(v) = 2\delta(v^2 -1) \theta (v^0)\, d^4 v$ is the invariant measure on
the unit hyperboloid. If $\Sc$ is not a Cauchy surface, then $\p(x)$ is
still uniquely determined by (\ref{cauchy}) in the
domain of causal dependence of ${\cal S}$.
Similarly the Fourier representation of the free Dirac
field can be written as
\be
\p(x) =\lp{m\over 2\pi}\rp^{3/2} \int e^{{\textstyle-im\,x \s v\,\g\s v}}
\g\s v\, f(v)\,\m(v)\,  ,
\label{freed}
\ee
with $f$ some complex 4-component function on the unit hyperboloid.
If we set $x=\la z$, with $z^2=1$, $z^0>0$ then the
leading asymptotic term when
$\la\rightarrow\infty$ is
$$
\p(\la z)\, \sim\, -i\,\la^{-3/2}
e^{{\textstyle -i( m\la + \pi /4)\g\s z}} f(z)\, .
$$
To see this one only has to observe that
${\dsp e^{{\textstyle -i\alpha\g\s v}} = e^{{\textstyle-i\alpha}} P_+(v) +
e^{{\textstyle+i\alpha}} P_-(v)}$, with \linebreak
\mbox{$P_\pm(v) = {1\over 2 }(1\pm \g\s v)$}, and use the standard stationary phase
method. We note for later use that $P_\pm^2 = P_\pm$, $P_+ P_- = P_- P_+ = 0$,
$P_+ + P_- = 1$. The above asymptotic behaviour of free field will
guide us to the reformulation mentioned at the beginning of this section.

We start with some geometric preliminaries. Let $x=\la z$, with
\mbox{$z^2=1$,} $z^0
>0$, and let $\n_a$ denote the flat derivative with respect to $x^a$.
We denote $\delta_a = \la (\n_a - z_a\ewl)$. $\delta_a$~is the derivative in the
directions tangent to the hyperboloid, and
${\dsp \lb\delta_a,\ewl \rb = 0}$. Moreover,
$\delta_a z^b = h_a^b$, where $h_a^b = g_a^b - z_a z^b$ is the projection tensor.
Every vector (and tensor) can be decomposed according to $\xi^a =
z^a \xi\s z + \xi_\T{}^a$, $~\xi_\T{}^a z_a =0$. In particular, the algebra of the
Dirac matrices is given by
\begin{eqnarray*}
(\g\s z)^2 &=&1,\\
\g\s z\, \g_\T{}^a + \g_\T{}^a\,\g\s z&=&0,\\
\g_\T{}^a\,\g_\T{}^b + \g_\T{}^b\,\g_\T{}^a&=&2h^{ab}.\\
\end{eqnarray*}
For all differentiable functions which fall off fast enough for the surface
term in the Stokes theorem
$$
0 = \int_{x^2 = 1}\n_c\{(x^cg_a^b - x^bg_a^c) f(x)\}\, d\sigma_b(x)
$$
to vanish, one has the integral identity involving only the $\delta_a$ -
derivative:
\be
\int(\delta_a - 3 z_a) f(z)\,\m(z) = 0.
\label{idi}
\ee

The Dirac equation
$$
\left[\g\s(i\n - eA(x)) - m\right] \p (x) = 0
$$
written in terms of the variables $\la$ and $z^a$ reads now
\be
i\ewl\, \chi (\la,z) = \{-\la^{-1} \g\s z\,\g_\T\s p + m\g\s z + e\g\s z\,
\g\s A(\la z)\}\chi(\la, z)\, ,
\label{dirac}
\ee
where $\chi(\la, z) = \la^{3/2} \p (\la z)$, and the operator
\be
p_a = i\Big(\delta_a + {1\over 2}\g\s z\, \g_\T{}_a -
{3\over 2} z_a\Big)
\label{p}
\ee
has been introduced. The conserved current of the Dirac equation is now
$\bp (x)\g^a\p(x) = \la^{-3} \bc\g^a\ch(\la, z)$, which, when integrated over
the hyperboloid $x^2=\la^2$, $x^0>0$,
gives the conserved quantity $\int\ov{\ch(\la, z)}\g\s z\,\ch(\la, z)
\,\m(z)$; bar over a $4$-component spinor function
denotes the usual Dirac conjugation. The integrand is
easily shown to be\linebreak $\dsp\frac{1}{z^0}(\ch_+^\dagger \ch_+
+ \ch_-^\dagger \ch_-)$, where the
dagger denotes the matrix hermitian conjugation and $\ch_\pm = P_\pm \ch$.
The quantity is thus positive definite, which suggests the precise formulation
of the problem as a unitary evolution in the Hilbert space $\Hc$ of the
equivalence classes of $\mathrm{\mathbf C}^4$ - valued functions on
the unit hyperboloid
$z^2 = 1$, $z^0 > 0$, with the scalar product
\be
(g, f) = \int \ov{g(z)}\g\s z f(z)\,\m(z)\, .
\label{scalar}
\ee
(We note that this cannot be achieved by a simple evolution-independent
unitary transformation within the usual formulation on hypersurface of constant
time $x^0$, as the change to the hyperboloid mixes the space and time aspects.)
Special classes (dense in $\Hc$)
of such functions such as k-times continuously differentiable
functions of compact support $C_0^k$ and the Schwartz test functions $\Sc$ are
defined as those $f(z^0, \vec{z})$ for which the respective properties hold
for $f(\sqrt{1 + \vec{z}^2}, \vec{z})$ with respect to $\vec{z}$; the
identification is time-axis independent. In the Hilbert space $\Hc$
the operator of multiplication by $\g\s z$ is easily seen to be a self-adjoint
unitary operator and $P_\pm$ become projection operators.
The operators $i\g_\T^a$, and $p_a$ defined in (\ref{p}), are not
bounded, but they are symmetric on each of the
special class of functions mentioned above.

The discussion of the free field case is best carried through with the use
of Fourier-type transformation on the unit hyperboloid. For functions in
$C_0^\infty $ we define two integral transformations
\begin{align}
 F_\kappa f(u)&=\lp{\kappa\over 2\pi}\rp^{3/2}\int e^{{\textstyle -i\kappa u\s z\, \g\s z}}\g\s z\,f(z)\, \m(z)\, ,\label{four1}\\[1ex]
 F_\kappa^* f(u)&=\lp{\kappa\over 2\pi}\rp^{3/2}\int e^{{\textstyle +i\kappa u\s z\,\g\s u}}\g\s z\,f(z)\, \m(z)\, .\label{four2}
\end{align}
By $F_\kappa f(u)$ and $F_\kappa^* f(u)$ we mean functions defined on
the unit hyperboloid, but the above integrals are valid outside the hyperboloid
as well.

\begin{pr}
$F_\kappa$ and $F_\kappa^*$ are isometric operators from $C_0^\infty$ into $\Sc$,
so they both can be extended to isometries of $\Hc$.
$F_\kappa$ and $F_\kappa^*$ are then mutually conjugated, hence they are unitary.
\label{four3}
\end{pr}
\begin{proof}
If $f\in C_0^\infty$ then $F_\kappa f(u)$ and $F_\kappa^* f(u)$ are infinitely
differentiable. Denote $\chi(u)$ the extension of $F_\kappa f(u)$ outside the
hyperboloid. $\chi(u)$ is a~regular wave packet, in the sense of ref.
\cite{ree}, so it vanishes rapidly between (and on) the surfaces $u^0=0$ and
$u^2=1$, $u^0>0$, in particular $F_\kappa f\in\Sc$. Moreover, the Dirac equation
$(i\,\g\s\n - \kappa )\chi(u)=0$ is satisfied. Hence, the current conservation
gives
$$
\int \ov{F_\kappa f(u)}\g\s u F_\kappa f(u)\,\m(u) =
\int \chi^\dagger(0, \vec{u})\,\chi(0, \vec{u})\,d^3 u.
$$
By the usual Fourier transform properties the r.h. side is easily transformed into
$$
\int {d^3z\over (z^0)^2}\lp f_+^\dagger(z)f_+(z) +
f_-^\dagger(z)f_-(z) -
f_+^\dagger(z)f_-(z^0, -\vec{z}) - f_-^\dagger(z)f_+(z^0, -\vec{z})\rp,
$$
which, after some manipulation with projectors $P_\pm$, gives
$\int \ov{f(z)}\,\g\s z\, f(z)\, \m(z)$.
This shows that $F_\kappa:C_0^\infty
\to\Sc$ isometrically, so it extends to isometry of $\Hc$. To prove the same
result for $F_\kappa^* f$ we assume that the support of $f$ lies in $z^0<a$.
For $z^0>a+\e$ deform the hyperboloid smoothly in such a way, that for large
$|\vec{z}|$ it tends to $z^0=a+2\e$, and regard $f(z)$ as initial data
on this surface for $k(z)$ satisfying \mbox{$(i\,\g\s\n -\kappa)\, k(z) = 0$}.
Then $k(z)$ has compact support between this surface and $z^0=0$, and
moreover
\mbox{$\n_a^{(z)}\big(e^{\textstyle i\kappa u\s z \g\s u} \g^a k(z)\big) = 0$.}
Therefore, changing the surface of integration in (\ref{four2})
to $z^0=0$ one obtains
$F_\kappa^* f(u)= P_+(u)\, G(\vec{u}) + P_-(u)\, G(-\vec{u})$, where
$$
G(\vec{u})= \lp\kappa\over 2\pi\rp^{3/2}\int
e^{{\textstyle -i\kappa\vec{u}\s\vec{z}}}
\g^0 k(0, \vec{z})\, d^3z\, .
$$
Hence $F_\kappa^* f(u)$ is a function of
fast decrease and one finds
$$
\int\ov{F_\kappa^* f(u)}\,\g\s u\,
F_\kappa^* f(u)\,\m(u) = \int G^\dagger (\vec{u})G(\vec{u})\, d^3 u\, .
$$
This, by
standard Fourier transformation properties and then by current conservation
for $k(u)$, is again $\|f\|^2$. Finally, one easily proves $(F_\kappa f, g)=
(f, F_\kappa^* g)$ for $f, g\in C_0^\infty$, which extends to $\Hc$. This ends the proof.
\end{proof}

The operator $U_0(\la_2, \la_1) = F_{m\la_2} F_{m\la_1}^*$ can be now
identified as the evolution operator of the free Dirac field in $\Hc$.
Indeed, the following proposition holds.

\begin{pr}
The families of operators $F_\kappa$, $F_\kappa^*$ and $U_0(\la_2, \la_1)$ are strong\-ly continuous in their parameters. For $f\in C_0^\infty$ the vectors
$F_\kappa^* f$ and $U_0(\la_2, \la_1)f$ are strongly differentiable in
$\kappa$, $\la_2$ and $\la_1$ according to the following formulae
\begin{align}
 -i{d\over d\kappa}F_\kappa^* f
 &= F_\kappa^*\,H_\kappa\, f\, ,\label{difk}\\[1ex]
 i\ewt U_0(\la_2, \la_1) f
 &= mH_{m\la_2} U_0(\la_2, \la_1)f\,,\label{difu2}\\[1.5ex]
 -i\ewo U_0(\la_2, \la_1) f
 &= U_0(\la_2, \la_1) mH_{m\la_1} f\, ,\label{difu1}
\end{align}
where $\dsp H_\kappa = \G \Big( -{1\over \kappa} \g_\T\s p + 1\Big)$\, ,
$\G f(z) = \g\s z f(z)$.
\label{free}
\end{pr}
\begin{proof}
If $f\in C_0^\infty$ then $F_{m\la_1}^*f\in\Sc$ and $F_{m\la_2}$ applied to the latter is therefore expressible in the integral form (\ref{four1}). The
function $F_{m\la_2}F_{m\la_1}^*f(z)$ is continuously differentiable in
$\la_2$ and $z$, satisfies in these variables the free version of
equation (\ref{dirac}), and for $\la_2=\la_1$ is equal to $f(z)$. Formulated
in terms of the original Dirac equation this means that\linebreak
$F_{m\la_2}F_{m\la_1}^*f(z) = \la_2^{3/2}\,\psi(\la_2 z)$, where $\psi(x)$
is the solution (\ref{cauchy}) of the initial data
problem for the Dirac equation
with the initial data $\psi(x)= \la_1^{-3/2}\, f(z)$ on $x=\la_1 z$. Both
$F_{m\la_2}F_{m\la_1}^*f(z)$ and its derivative on $\la_2$ are therefore
jointly continuous in $\la_2$ and $z$, and have compact support in $z$ for
$\la_2$ in some neighbourhood of $\la_1$. The strong continuity of
$U_0(\la_2, \la_1)$, differentiability of $U_0(\la_2, \la_1) f$ in
$\la_2$ and eq.(\ref{difu2}) now easily follow.
$F_\kappa = U_0(\kappa/m, \kappa_0/m) F_{\kappa_0}$ is then strongly
continuous as well, so as is
$U_0(\la_2, \la_1)$ in $\la_1$. Strong differentiability of
$U_0(\la_2, \la_1)f$ in $\la_1$ for $f\in C_0^\infty$ and eq.(\ref{difu1}) follow from
\begin{eqnarray*}
&&\left\|\frac{U_0(\la_2, \la'_1)f - U_0(\la_2, \la_1)f}{\la'_1 - \la_1}
- U_0(\la_2, \la_1) imH_{m\la_1} f\right\| \\
&&=\left\|\frac{U_0(\la'_1, \la_1)f - f}{\la'_1 - \la_1}
+ U_0(\la'_1, \la_1) imH_{m\la_1} f\right\| \to 0
\end{eqnarray*}
for $\la'_1\to\la_1$. In consequence (\ref{difk}) follows as well.
\end{proof}

We are now ready to discuss evolution in the presence of the electromagnetic
potential $A_a(x)$. Let $R(\la_2, \la_1)$ be the unitary propagator generated
by the family of operators $V_R (\la) = F_{m\la}^*V(\la)F_{m\la}$, where
$V(\la)$ is the operator of multiplication by $e\g\s z\,\g\s A(\la z)$.
For $R(\la_2, \la_1)$ to be well defined it suffices to assume (which is
sufficient for our purposes), that $V(\la)$ is a strongly continuous family
of bounded operators. Then $R(\la_2, \la_1)$ is jointly strongly continuous
in $\la_2$ and $\la_1$ and
\begin{align}
 i\ewt R(\la_2, \la_1)f &= V_R(\la_2)\, R(\la_2, \la_1)f\, ,\label{r2}\\[1ex]
 -i\ewo R(\la_2, \la_1)f &= R(\la_2, \la_1)V_R(\la_1)f\,\label{r1}
\end{align}
for any $f\in\Hc$. If $a_a(z)$ is a measurable vector function then
$\|\g^\T\s a\, f\| = \| \sqrt{-a_\T^2} f\|$, hence
$\|\g\s a\, f\|\leq \|a\s z f\| +\|\sqrt{-a_\T^2} f\|$.
To satisfy the conditions on $V(\la)$ we assume therefore that $A_a(x)$ is
continuous and both $|z\s A(\la z)|$ and $|A_\T^2(\la z)|$ have bounds
independent of $z$.

The unitary propagator $U(\la_2, \la_1) = F_{m\la_2} R(\la_2, \la_1)
F_{m\la_1}^*$ gives the Dirac evolution at least in the weak sense
$$
i\ewt (f, U(\la_2, \la_1) g) = ([mH_{m\la_2} + V(\la_2)]f, \,
U(\la_2, \la_1)g)\, ,
$$
where $g$ is any vector in $\Hc$ and $f\in C_0^\infty$. Moreover, for any $f\in C_0^\infty$
$$
-i\ewo U(\la_2, \la_1)f = U(\la_2, \la_1)(mH_{m\la_1} + V(\la_1)) f\, .
$$

The scattering states of the evolution so determined are easily obtained
by a simple unitary transformation, as suggested by the asymptotics of
$\psi(\la z)$ discussed at the beginning of this section. Let us denote
$G_\kappa = e^{\dsp i\kappa\G}$, where $\G$ is the operator defined in
Proposition \ref{free}. $G_\kappa$ is strongly continuous, differentiable on
every $f\in\Hc$, family of unitary operators. Denote further
$T_\kappa = iG_{\kappa + \pi /4} F_\kappa$ and $W(\la_2, \la_1) = G_{m\la_2 + \pi/4}
U(\la_2, \la_1)G_{m\la_1 + \pi/4}^* = T_{m\la_2} R(\la_2, \la_1)
T_{m\la_1}^*$.

\begin{lem}
If $f\in C_0^2$ then $\dsp \| T_\kappa f - f\| = \|T_\kappa^* f - f\| \leq
{1\over \kappa} (\| hf\| + \| h^2f\|)$, where $h= -i\g_\T\s p$.
\label{ast}
\end{lem}
\begin{proof}
Let $f\in C_0^\infty$ first. Then $G_\kappa^* f \in C_0^\infty$ as well, so the
differentiations in $\dsp {d\over d\kappa}T_\kappa^* f$
can be performed. Using the fact
that $p$ commutes and $\g^\T$ anticommutes with $\G$ one finds
$$
{d\over d\kappa}T_\kappa^* f = T_\kappa^* {i\over\kappa} G_{2\kappa} hf =
T_\kappa^* {d\over d\kappa} g_{2\kappa} hf\, ,
$$
where ${\dsp g_\kappa = -i\int_\kappa^\infty G_u
{du\over u} = \G \lp{G_\kappa\over \kappa}-\int_\kappa^\infty G_u {du\over u^2}\rp}$
(the first form as an improper integral). From the latter form one has
$\dsp \|g_\kappa\|\leq {2\over\kappa}$. $g_{2\kappa}hf$ is again in $C_0^\infty$, so
$$
{d\over d\kappa}\lp T_\kappa^* f - T_\kappa^* g_{2\kappa} hf\rp =
T_\kappa^* {i\over \kappa} G_{2\kappa} g_{2\kappa}^* h^2f\, .
$$
Integration from $\kappa_1$ to $\kappa_2$ leads to
$$
T_{\kappa_2}^*f - T_{\kappa_1}^*f = T_{\kappa_2}g_{2\kappa_2}hf - T_{\kappa_1}g_{2\kappa_1}hf +
i\int_{\kappa_1}^{\kappa_2} T_u^* G_{2u} g_{2u}^* h^2 f {du\over u}
$$
and
$$
\|T_{\kappa_2}^* f - T_{\kappa_1}^* f\|\leq \lp {1\over\kappa_2} + {1\over\kappa_1}\rp
(\|hf\| + \|h^2 f\|)\, .
$$
This shows that $T_\kappa^* f$ has a limit; this limit has to be $f$, as for any
$g\in C_0^\infty$ there is $(g, T_\kappa^* f - f) \rightarrow 0$, which is easily seen
e.g. by stationary phase method. Taking the limit $\kappa_1 \rightarrow \infty$
one obtains the stated result for $f\in C_0^\infty$. Any $f\in C_0^2$ can be uniformly
approximated together with its derivatives by functions from $C_0^\infty$ vanishing
outside a common compact set. This ends the proof.
\end{proof}

The above lemma reduces the problem of asymptotics of $W(\la_2, \la_1)$ to
that of $R(\la_2, \la_1)$
\begin{equation}\label{w1}
\begin{aligned}
 \|W(\la, \la_2)f &- W(\la, \la_1)f\|\\
 &\leq \| R(\la, \la_2)f -R(\la, \la_1)f\|
+ \Big({1\over m\la_2} + {1\over m\la_1} \Big) \lp\|hf\| + \|h^2f\|\rp
\end{aligned}
\end{equation}
for $f\in C_0^\infty$. The generator of the propagator $R(\la_2, \la_1)$ can be written in the form \mbox{$V_R(\la) = T_{m\la}^* \lp G_{2m\la} v_1(\la) +
v_2(\la)\rp T_{m\la}$}, where $v_1(\la)$ and $v_2(\la)$ are the operators
of multiplication by $ie\g_\T\s A(\la z)$ and $ez\s A(\la z)$ respectively.
Transforming (\ref{r1}) with the use of the method applied in the proof of
Lemma~\ref{ast} one obtains
\begin{equation}\label{r3}
\begin{aligned}
 -i\partial_u &\left\{ R(\la, u)f - R(\la, u)g_{2mu} uv_1(u)f\right\}
 =R(\la, u)T_{mu}^*v_2(u)T_{mu}f \\[1ex]
 &-R(\la, u) V_R(u) g_{2mu} uv_1(u)f + R(\la, u)T_{mu}^*G_{2mu}v_1(u)(T_{mu} - \I )f\\
 &+R(\la, u)(T_{mu}^* - \I)G_{2mu}v_1(u)f
 + R(\la, u)g_{2mu} i{d\over du} uv_1(u)f\, .
\end{aligned}
\end{equation}
The strong differentiation in the last term will be allowed under the
assumptions of the following theorem.

\begin{pr}
Let $A_a(x)$ be a vector function twice continuously differentiable and
for $\la > \la_0 > 0$ subject
to the following bounds for some $\e > 0$
\begin{eqnarray}
& &|z\s A(\la z)|< \frac{\con}{\la^{1 + \e}},~~~~~ |A_\T^2 (\la z)| <
\frac{\con}{\la^{1 + \e}},\nonumber\\
& &\left| \left[\ewl (\la A_\T(\la z))\right]^2\right| < {C(z)\over \la^{2\e}},
\label{bounds}\\
& &|\delta_aA_\T{}_b(\la z)| < {D(z)\over \la^\e}, ~~~~~
|\delta_a\delta_bA_\T{}_c(\la z)| < {D(z)\over \la^\e}, \nonumber
\end{eqnarray}
where $C(z)$ and $D(z)$ are continuous functions and the last two estimates
hold component-wise in arbitrary fixed Lorentz frame (change of the frame
results only in the change of $D(z)$).

Then for all $f\in C_0^\infty$
\be
\|W(\la, \la_2)f - W(\la, \la_1)f\| \leq c(f)\left( {1\over \la_2^\alpha} +
{1\over \la_1^\alpha}\right),\label{asw}
\ee
where $\alpha = \min\{\e, 1\}$ and $c(f)$ is a constant depending on $f$.
Hence for every $\la>\la_0$ the strong limit
${\dsp \lim_{u\to\infty} W(\la, u)f = f_\la }$ exists, is strongly
continuous in $\la$ and
\be
\|f_\la - f\| \leq {c(f)\over \la^\alpha}.
\label{asf}
\ee
\label{scat1}
\end{pr}
\begin{proof}
To prove (\ref{asw}) it remains to estimate various terms in
(\ref{r3}). The successive terms (i),...,(v) on the r.h. side of (\ref{r3})
are bounded in norm respectively by
\begin{eqnarray*}
& &\|(i)\| \leq \|v_2(u)\|\,\|f\|\, ,\\
& &\|(ii)\| \leq {1\over m}(\|v_1(u)\| + \|v_2(u)\|)\, \|v_1(u)\|\,\|f\|\, ,\\
& &\|(iii)\| \leq \|v_1(u)\|{1\over mu}(\|hf\| + \|h^2f\|)\, ,\\
& &\|(iv)\| \leq {1\over mu}(\|hv_1(u)f\| + \|h^2v_1(u)f\|)\, ,\\
& &\|(v)\| \leq {1\over mu}\left\| {d\over du} uv_1(u)f\right\|\, ;
\end{eqnarray*}
for (iv) the fact was used, that $v_1(u)f\in C_0^2$. The assumed estimates
of the potential force all these bounds below some constant depending on $f$
times $\la^{-1 -\alpha}$. The integration of (\ref{r3}) leads therefore to
$$
\left\|\left[ R(\la, u)f - R(\la, u)g_{2mu} uv_1(u)f\right]
\Big|_{u=\la_1}^{u=\la_2}\right\| \leq \con (f)
\left( {1\over \la_2^\alpha} + {1\over \la_1^\alpha}\right)
$$
The form of this inequality remains unchanged, if we omit the second
term inside the brackets on the l.h. side (this term is bounded by
$\con \|f\|\, u^{-1-\e}$). Taking into account (\ref{w1}) one arrives at
(\ref{asw}). The continuity of $f_\la$ is evident from $f_{\la'} =
W(\la', \la)f_\la$, and (\ref{asf}) is obtained by putting $\la_1=\la$
and letting $\la_2\rightarrow\infty$ in (\ref{asw}). This ends the proof.
\end{proof}

\begin{col}
For every $f\in\Hc$ the strong limit
${\dsp \lim_{u\to\infty} W(\la, u)f = f_\la }$ exists.
$f_\la$ is strongly continuous and $\|f_\la - f\|\rightarrow 0$
for $\la\rightarrow\infty$.
\label{scat2}
\end{col}

The crucial point of our discussion is the fact, that the long-range
electromagnetic fields of the Coulomb type are admitted by the premises
of Proposition \ref{scat1}, provided one chooses the potential
in an appropriate gauge. Let us observe first, that if the electromagnetic
field is represented as a superposition, it suffices to satisfy
(\ref{bounds}) for the potentials of the superposed fields separately.
Suppose that one of the superposed fields is the asymptotic Coulomb-type
field homogeneous of degree $-2$: $F_{ab}(\kappa x) = \kappa^{-2} F_{ab}(x)$.
The simplest choice of the potential inside the lightcone is of the form
$A_a(\la z) = \la^{-1} a_a(z)$. This potential breaks the first of the bounds
(\ref{bounds}). Assume, however, that $a_a (z)$ is three times continuously
differentiable and satisfies the bounds
\be
\left| a_\T^2(z)\right| < \con\, ,
~~~\left|\left[\delta(z\s a(z))\right]^2\right| < \con\, .
\label{abounds}
\ee
Choose the new gauge by $A_{\mrm{tr}}(x) = A(x) - \n S(x)$ with $S(x)$ given
by $S(\la z) = \ln\la\, z\s a(z)$ inside the lightcone. Then
\be
A_{\mrm{tr}}{}_b(\la z) = \la^{-1} \left\{a_\T{}_b (z) - \ln\la\,
\delta_b(z\s a(z))\right\}
\label{atr}
\ee
and $z\s A_{\mrm{tr}}(\la z) = 0$. The other bounds of (\ref{bounds}) are
satisfied for any $\e < 1$ (with constants depending on $\e$).

Another class of potentials admitted by Proposition (\ref{scat1})
consists of Lorentz-gauge potentials (\ref{freepot1}) of free fields
discussed in Section \ref{selm}. With the use of (\ref{bl1}, \ref{bl2})
we get
\[
|A_a(\la z)|<\frac{\con}{\la z^0}\, ,\quad
|\n_aA_b(\la z)|<\con\frac{(z^0)^\e}{\la^{2+\e}}\, ,\quad
|\n_a\n_bA_c(\la z)|<\con\frac{(z^0)^\e}{\la^{2+\e}}\, .
\]
These bounds imply the third, fourth and fifth of
the estimates (\ref{bounds}), while the first bound above is sufficient for
the second estimate in (\ref{bounds}) to hold, if the first one is satisfied.
To prove this remaining estimate we observe first that, as follows from
(\ref{freepot2}), $V_a(s, l)= o_A k_{A'}(s, o, \oc) + \cc$, where
$o_{C'}k^{C'}(s, o, \oc)=\z(s, o, \oc)$. Inserting this into
(\ref{freepot1}) we get by (\ref{inth2})
$$
x\s A(x)=\frac{1}{2\pi}\int\D'_{A'}k^{A'}(x\s l. o, \oc)\, \dl + \cc\, .
$$
Hence $\dsp |z\s A(\la z)|
<\frac{\con}{\la^{1+\e}(z^0)^\e}$ by (\ref{bl0}),
which ends the proof.

 We stress that the transformation used here to compensate the asymptotic
behaviour of the Dirac field is interaction-independent, unlike in the
usual Dollard treatment of the Coulomb potential \cite{dol}, \cite{ree},
or in a recent discussion of the Cauchy problem for the classical spinor
electrodynamics \cite{fla}.

The Dirac field is expressed in terms of
$f_\la(z)$ as
\be
\p(\la z) = -i\, \la^{-3/2}\, e^{\dsp -i(m\la + \pi /4)\g\s z}
f_\la (z)\, .
\label{psif}
\ee
If $f_\la^0 (z)$ is a solution of the free evolution, with the corresponding
Dirac field $\p^0(\la z)$, then
$$
\int_{x^2 =\la^2} \bp^0(x)\g^a\p(x)\, d\sigma_a(x) = (f_\la^0, f_\la)
\rightarrow (f^0, f)
$$
for $\la\to\infty$. This suggests that the precise formulation of the
asymptotic Dirac field in the quantum electrodynamics be looked for as
a limit of the expression on the l.h. side, with $\p^0$ being a test field.

\setcounter{equation}{0}
\setcounter{pr}{0}
\section{Total conserved quantities}
\label{stot}

We want to return now to the consideration of a closed system with
electromagnetic interaction, which has been taken up in Section \ref{srad}.
The results should not depend essentially on what kind of massive
field one couples minimally to the electromagnetic field, but we consider
for definiteness the Maxwell-Dirac system. The discussion of the present
section will make use of the results of the preceding sections, but in the
full theory we lack rigorous results along the lines presented here.
Rigorous results on the Cauchy problem and scattering properties of
the Maxwell-Dirac theory were recently reported by Flato et al \cite{fla}, but
the method used by these authors is quite different, and the relation
of the present work with \cite{fla} remains to be clarified.
What is clear, however, is the difference in the choice
of transformation leading to asymptotic states: our transformation is
interaction-independent, which is made possible by a special choice of gauge,
while the transformation of Flato et al is a Dollard-type treatment (cf.
\cite{dol}), consisting of extraction of a phase in momentum space, thus not
constituting a gauge transformation in the usual sense. Moreover, the
method used in the present work aims at appropriate description of the
spacetime separation of asymptotic matter and radiation, so far as it can
be achieved. We stress, however, that no results on the Cauchy problem or
asymptotic completeness are given here.

Proceeding heuristically we shall assume that the asymptotics of fields
of the interacting theory are of the type described in Section \ref{selm}
for the electromagnetic and in
Section \ref{sdir} for the Dirac field respectively.
When needed we shall add further assumptions on how these asymptotics are
achieved. These extrapolations seem plausible, provided (i) the full
electromagnetic potential falls into the class admitted by Proposition
\ref{scat1} and (ii) the current of the Dirac field vanishes in spacelike
directions sufficiently fast for the discussion of Section~\ref{selm}
to remain valid. Basing the intuitions on the free Dirac field case we
regard the second point as unproblematic, but for its rigorous justification
more control over the limit $f_{\la} \to f$, and also the solution
of the Dirac equation outside the cone would be needed. As to the first
point, we can only present a very simplified argument of self-consistency type,
which, however, takes care of the Coulomb term, the most troublesome from
the point of view of asymptotics of the matter field.

More explicitly, we represent the Dirac field inside the lightcone
as in (\ref{psif}) and assume that $f_{\la}\to f$ as in Corollary \ref{scat2}.
For any current density denote inside the lightcone
$j_a(\la, z) = \la^3 J_a(\la z)$, $z^2 =1$, $z^0>0$. For the Dirac field
\be
j_a(\la, z)= z_a\rho(z) + \lp e^{\textstyle -2im\la}\kappa_a(z) + \cc \rp
+ r_a(\la, z)\, ,
\label{jot}
\ee
where
\begin{align*}
 \rho(z) &= e\,\ov{f(z)}\g\s z f(z)\, ,\\
 \kappa_a(z) &= -ie\,\ov{P_- f(z)}\,\g_\T{}_a\,P_+f(z)\, ,\\
 r_a(\la, z) &= e\,\ov{e^{\textstyle -i(m\la+ \pi/4)\g\s z}f_\la(z)}\,\g_a\,
 e^{\textstyle -i(m\la+ \pi/4)\g\s z}f_\la(z)\\
 &-e\,\ov{e^{\textstyle -i(m\la+ \pi/4)\g\s z}f(z)}\,\g_a\,
 e^{\textstyle -i(m\la+ \pi/4)\g\s z}f(z)\, .
\end{align*}
The electromagnetic potential in the Lorentz gauge can be split into the free
outgoing and advanced parts. As for the free part, its admissibility in
Proposition \ref{scat1} has been proved already in the preceding section.
The advanced field of the current $J_a(x)$ can be written inside the future
lightcone as
$$
\Aa{}_a(x) = \int j_a\lp x\s v + \sqrt{(x\s v)^2 - x^2}, v \rp
\frac{\m(v)}{\sqrt{(x\s v)^2 - x^2}}\, .
$$
For the Dirac density the first term of (\ref{jot}) yields a Coulomb potential
\be
\Ac{}_b(\la z) = \frac{a_b(z)}{\la}\, ,
\label{poco1}
\ee
with
\be
a_b(z) = \int v_b\,\rho(v)\frac{\m(v)}{\sqrt{(z\s v)^2 - 1}}\, .
\label{poco2}
\ee
This is a homogeneous potential of the type discussed after Corollary
\ref{scat2}. All we have to show for admissibility of its gauged
form $A^{\mrm{Coul}}_{\mrm{tr}}$ (\ref{atr}) is the threefold
differentiability of $a_b(z)$ and the bounds (\ref{abounds}).
The differentiability of $a_b(z)$ follows easily by the use of the identity
$$
\delta_b\int h(v) \frac{\m(v)}{\sqrt{(z\s v)^2 - 1}} =
\int \lb z\s v (\delta_b - 3v_b) + z_b\rb h(v)\frac{\m(v)}{\sqrt{(z\s v)^2 - 1}}
$$
and suitable assumptions on the regularity and fall-off of $\rho(v)$.
(The identity follows by multiplication of
$$
\delta^{(z)}_b\frac{1}{\sqrt{(z\s v)^2 - 1}} =
\frac{z_b}{\sqrt{(z\s v)^2 - 1}} - \delta^{(v)}_b
\frac{z\s v}{\sqrt{(z\s v)^2 - 1}}
$$
by $h(v)$ and integration by parts according to (\ref{idi}).)
From (\ref{mest}) we have
\be
\left| a_b(z)\right| < \frac{\con}{z^0}\, ,
~~~~~\left| z\s a(z)\right| < \con\, ,
\label{poco3}
\ee
so that the first of the bounds (\ref{abounds}) is satisfied. To obtain the
other one we observe first that the components of any unit vector orthogonal to
a timelike unit vector $z^a$ are bounded by $z^0$, in particular
$\dsp \left|\frac{v^a - z\s v z^a}{\sqrt{(z\s v)^2 - 1}}\right|
\leq z^0$. (Proof: if $w\s z =0$, then $|z^0w^0|\leq|\vec{z}||\vec{w}|$, or,
using $z^2=-w^2=1$, $(|\vec{w}|^2 -1)(z^0)^2\leq |\vec{w}|^2((z^0)^2-1)$,
hence $z^0\geq |\vec{w}|\geq |w^0|$.) Hence, by (\ref{mest}),
\be
\delta_b(z\s a(z)) = -\int \rho(v)\,\frac{v_b- z\s v z_b}{\sqrt{(z\s v)^2 - 1}}\,
\frac{\m(v)}{(z\s v)^2 -1}\, ,
\label{ga}
\ee
and $\dsp \left|\delta_b(z\s a(z))\right|<\frac{\con}{z^0}$,
which implies the second
of the inequalities (\ref{abounds}). From now on $A^{\mrm{Coul}}_{\mrm{tr}}$
replaces $A^{\mrm{Coul}}$ in the Dirac equation.

The remaining contributions to $\Aa{}_b(\la z)$ will not be discussed in
detail, but we assume, what could be achieved with
some additional assumptions on uniformness of the limit
$f_\la(z)\to f(z)$ and on regularity of $f(z)$, that
\be
\left|\Aa{}_b(\la z) - \Ac{}_b(\la z)\right| < \frac{\con}{(\la z^0)^{1+\e}}\,
\label{rest1}
\ee
and
\be
\left|F^{\mrm{adv}}{}_{ab}(\la z) - F^{\mrm{Coul}}{}_{ab}(\la z)\right| <
\frac{\con}{(\la z^0)^{2+\e}}\, ,
\label{rest2}
\ee
where
\be
F^{\mrm{Coul}}{}_{ab}(\la z) = \frac{f_{ab}(z)}{\la^2}\, ,
\label{co1}
\ee
\be
f_{ab}(z) = \int\rho(v)\frac{z_av_b - z_bv_a}{\sqrt{(z\s v)^2 - 1}}\,
\frac{\m(v)}{(z\s v)^2 -1}\, .
\label{co2}
\ee
Since $z_{[a}v_{b]}=(z_{[a}-z\s v v_{[a}) v_{b]}$ we have
$\dsp \left|\frac{z_av_b-z_bv_a}{\sqrt{(z\s v)^2 - 1}}\right|
<2(v^0)^2$ and from (\ref{mest})
\be
\left| f_{ab}(z)\right|< \frac{\con}{(z^0)^2}\, ,
\label{co3}
\ee
if $\dsp |\rho(v)|<\frac{\con}{(v^0)^{4+\e}}$.

The new gauge of the electromagnetic potential, which we use here for
its simplicity, is a nonlocal one, being reached from a Lorentz gauge
by a transformation depending on the asymptotic current. However, the
same asymptotic effect can be achieved by a local gauge transformation
$A(x)\to A(x) - \n S(x)$, with $\dsp S(x) =
\ln\sqrt{x^2}\, x\s A(x)$ inside the lightcone.

We come now to our principal aim in this section. We want to complete the
discussion of Section \ref{selm} by supplying the up to now lacking
expressions for energy-momentum and angular momentum going out in
timelike directions with the massive part of the system. We recall, that these
quantities are determined by (\ref{outm1}) and (\ref{outm2}) respectively,
and they do not depend on the choice of the time-axis along which the limits in
those formulae are achieved. We take advantage of this independence to chose an
axis going through the origin of Minkowski space (with arbitrary time-vector
$t$). The total energy-momentum tensor of the theory is given by
$$
T_{ab} = T^{\mrm{D}}{}_{ab} + T^{\mrm{elm}}{}_{ab}\, ,
$$
where $T^{\mrm{elm}}{}_{ab}$ is the tensor of the total electromagnetic
field (\ref{emt}) and
$$
T^{\mrm{D}}{}_{ab} = \frac{1}{4} \left\{\ov{\psi}\g_a(i\n_b -eA_b)\psi
+ \cc \right\} + (a\leftrightarrow b)\, ,
$$
where $(a\leftrightarrow b)$ stands for terms with interchanged indices.
Recalling result (\ref{freefut}) we see that the contribution to the r.h.
sides of (\ref{outm1}) and (\ref{outm2}) coming from the out field vanish.
Also the contributions coming from the mixed adv-out terms in
$T^{\mrm{elm}}{}_{ab}$ vanish, as shown in \ref{apb}.

We are left with the task of calculating the r.h. sides of (\ref{outm1}) and
(\ref{outm2}) for
$$
T'_{ab} = T^{\mrm{D}}{}_{ab} + T^{\mrm{adv}}{}_{ab}\, ,
$$
where $T^{\mrm{adv}}{}_{ab}$ is the electromagnetic tensor of advanced field.
We want to show first that the limits of the
integrals over $\Co^{\mrm{fut}}(\tau)$ for $\tau\to\infty$
may be replaced by the limits for $\la\to\infty$ of
the integrals over hyperboloids  \mbox{$\Hc(\la)=\{x| x^2=\la^2, x^0>0\}$}.
To this end
consider integrals over the region contained between $\Co^{\mrm{fut}}(\tau)$
and $\Hc(\la)$ of the quantities
\be
\n^c T'_{ac} = -F^{\mrm{out}}{}_{ac}J^c
\label{dyw1}
\ee
and
\be
\n^c(x_a T'_{bc} - x_b T'_{ac}) = \lp -x_a F^{\mrm{out}}{}_{bc} +
x_b F^{\mrm{out}}{}_{ac}\rp J^c \, .
\label{dyw2}
\ee
Since $T'_{ab}$ gives no flow of energy-momentum or angular momentum
to null infinity, these integrals give the differences of energy-momentum and
angular momentum passing through $\Co^{\mrm{fut}}(\tau)$ and $\Hc(\la)$.
If the above divergencies are absolutely integrable over the region
$x^2 > 1$, $x^0>0$,
then these differences vanish in the limit and the replacement
of $\Co^{\mrm{fut}}(\tau)$ by $\Hc(\la)$ is justified. If we assume that
$\left| j_a(\la, z)\right| < h(z)$, then by (\ref{bl2})
the r.h. side of (\ref{dyw1})
is bounded by $\dsp
\con\frac{h(z) (z^0)^{1+\e}}{\la^4 z^0 (\la + s_t z^0)^{1+\e}}$ and
the r.h. side of (\ref{dyw2}) by a similar quantity multiplied by $\la z^0$.
For $x=\la z$ there is $d^4x = \la^3d\la\, \m(z)$,
so both these expressions are
integrable over $\la>1$ for $h(z)$ such that
$\dsp \int h(z) (z^0)^{1+\e}\, \m(z) <\infty$.

The preceding discussion brings us to the following representations
\begin{align*}
 P^{\mrm{out-t}}{}_a
 &= \lim_{\la\to\infty}\int \la^3 T'_{ac}(\la z) z^c\, \m(z)\, ,\\
 M^{\mrm{out-t}}{}_{ab}
 &= \lim_{\la\to\infty}
 \int \la^4\lp z_a T'_{bc}(\la z) - z_b T'_{ac}(\la z)\rp z^c\, \m(z)\, .
\end{align*}
The limits here will be treated rather formally, by assuming that for
large $\la$ only the leading (constant at least) terms of the integrands
contribute. In this way there is no contribution from $T^{\mrm{adv}}{}_{ab}$
to $P^{\mrm{out-t}}{}_a$ and contribution to $M^{\mrm{out-t}}{}_{ab}$ comes
from\\ $-\frac{1}{16\pi}\lp z_a f_{bd}(z) - z_b f_{ad}(z)\rp
f_c{}^d z^c$. This term, however, vanishes identically, since
\mbox{$z_{[a} f_{bc]} =0$}.

Consider finally $T^{\mrm{D}}{}_{ab}$, which gives the only nonvanishing
contributions. Writing $\psi(\la z) = \la^{-3/2}\ch(\la, z)$ and using the
Dirac equation (\ref{dirac}) we have
\begin{align*}
 (i\n^a &- eA^a)\psi(\la z) \\
 &= \la^{-3/2} \Big\{ z^a\g\s z
 \Big(m -\frac{1}{\la}\g_\T\s p + e\g\s A\Big) + \frac{1}{\la}p^a - eA^a
-\frac{i}{2\la}\g\s z \g_\T{}^a\Big\}\ch(\la, z)\, \
\end{align*}
where $p^a$ is the operator defined in (\ref{p}). Now, \\
\mbox{$\ch(\la, z)= -i e^{\textstyle -i(m\la+\pi/4)\g\s z} f_\la(z)$}
and we treat $\ch$ as $O(\la^0)$. Then
\begin{align*}
 &\la^3T^{\mrm{D}}{}_{ac} z^c = m z_a \ov{\ch}\ch + O(\la^{-\e})\, , \\
 &\la^4\lp z_aT^{\mrm{D}}{}_{bc} - z_bT^{\mrm{D}}{}_{ac}\rp z^c
 =\frac{1}{2}\left\{ z_a\ov{\ch}\,\g\s z \lp p_b-e\la A_b+\frac{1}{2}
 \g_\T{}_{[b}\g_\T{}_{c]} p^c\rp\ch - (a\leftrightarrow b)\right\}\\
 &\hspace{20em}+\cc \, .
\end{align*}
The result of integration over the hyperboloid can be written in terms of
the scalar product of Section \ref{sdir}
\begin{align}
 &\int\la^3 T^{\mrm{D}}{}_{ac}(\la z)z^c\, \m(z)
 = m(\ch, \g\s z z_a\ch) + O(\la^{-\e})\, ,
 \label{intdir1}\\
 &\begin{aligned}
 \int\la^4&\lp z_aT^{\mrm{D}}{}_{bc}(\la z)
 -z_bT^{\mrm{D}}{}_{ac}(\la z)\rp z^c\, \m(z) \\
 &= \lp\ch, (z_ap_b - z_bp_a)\ch\rp
 -e\lp\ch, (z_a\la A_b(\la z) - z_b\la A_a(\la z))\ch\rp \\
 &+ \frac{1}{4}\lp\ch,\left[ z_a\g_\T{}_{[b}\g_\T{}_{c]} -
 z_b\g_\T{}_{[a}\g_\T{}_{c]}, p^c\right]\ch\rp \, ,
\end{aligned}
\label{intdir2}
\end{align}
where the symmetry of operators was taken into account. The operators
appearing in the averages
commute with $\g\s z$, so $\ch$ can be replaced by$f_\la$, and further, up to
$O(\la^{-\e})$, by $f(z)$. Using
$\dsp \left[\g_\T{}_{[b} \g_\T{}_{c]}, p^c\right]=0$ and
$\left[p_c, z_a\right]= ih_{ca}$ we transform the third term in (\ref{intdir2})
to the form $\dsp \frac{i}{4}\lp f,
\left[\g_\T{}_a\g_\T{}_b\right]f\rp
+ O(\la^{-\e})$. Contributions to the second term up to $O(\la^{-\e})$
could only come from $A^{\mrm{Coul}}_{\mrm{tr}}$. However,
\begin{eqnarray*}
&&-2e\lp f, z_{[a}\lp a_{b]}(z) - \ln\la\, \delta_{b]}(z\s a(z)\rp f\rp \\
&&=\int\rho(z)\int\rho(v)\frac{z_av_b - z_bv_a}{\sqrt{(z\s v)^2 -1}}
\lp 1 + \frac{\ln\la}{(z\s v)^2-1}\rp\, \m(v)\, \m(z) = 0\, ,
\end{eqnarray*}
due to antisymmetry of the integrand with respect to interchange of
integration variables $z\leftrightarrow v$. Taking now the limit $\la\to\infty$
we finally obtain
\begin{eqnarray}
&&P^{\mrm{out-t}}{}_a = m\lp f, \g\s z z_af\rp=
m\int z_a \ov{f}f(z)\,\m(z)\, ,
\label{poutt}\\
&&\begin{array}{@{}l}
\dsp M^{\mrm{out-t}}{}_{ab}= \lp f, \lp z_ap_b-z_bp_a
+\frac{i}{4}\lb \g_\T{}_a, \g_\T{}_b\rb\rp f\rp \\
\dsp\hspace*{2cm} =\int \ov{f}\g\s z\lp z_ai\delta_b-z_bi\delta_a +
\frac{i}{4}\lb\g_a, \g_b\rb\rp f(z)\, \m(z)\, ,
\end{array}
\label{moutt}
\end{eqnarray}
which are the desired formulae for the quantities going out in timelike
directions. If we define the free outgoing Dirac field by (cf. (\ref{freed}))
$$
\psi^{\mrm{out}}(x)= \lp\frac{m}{2\pi}\rp^{3/2}\int
e^{\textstyle -imx\s v \g\s v}\g\s v\, f(v)\, \m(v)\, ,
$$
then the above expressions give the Fourier representations of the
conserved quantities of this field (in a somewhat unusual but most compact
form). Similar expressions could be obtained for the timelike past infinity.

The task of expressing the total energy-momentum and angular momentum of the
interacting theory in terms of asymptotic fields has been now completed.
As anticipated, the contributions of electromagnetic and massive free
fields almost separate, except for a term in the radiated angular momentum
due to the long-range part of the electromagnetic field. This term
(the second one on the r.h. side of (\ref{split2})) can be now rewritten
by the use of matter asymptotics. $q(o, \oc)$ is now given by formula
(\ref{Di6}), hence
\begin{align*}
 \Delta\mu_{AB}&\equiv \frac{1}{2\pi}\int qo_{(A}\D_{B)}\F(o, \oc)\, \dl
 =-\frac{e}{4\pi}\int \ov{f}\g\s z f(z)
 \int \F(o, \oc)o_{(A}\D_{B)}\frac{1}{(z\s l)^2}\, \dl\, \m(z)\\
 &= -\int \ov{f}\g\s z f(z)(z) z_{C'(A}\delta_{B)}^{C'} H(z)\,\m(z)\, ,
\end{align*}
or in the tensor form
$$
\Delta M_{ab}\equiv \Delta\mu_{AB}\e_{A'B'} + \cc =
-(f, (z_a\delta_bH - z_b\delta_aH)f)\, ,
$$
where
\be
H(z)=\frac{e}{4\pi}\int\frac{\F(l)}{(z\s l)^2}\, \dl\, .
\label{phase}
\ee
If we now change the phase of $f(z)$ by introducing
\be
g(z)= e^{\textstyle iH(z)} f(z)\, ,
\label{chph}
\ee
then $M^{\mrm{out-t}}{}_{ab} + \Delta M_{ab}$ has again the form
(\ref{moutt}) but with $f$ replaced by $g$, while $P^{\mrm{out-t}}{}_a$
retains its form under this replacement. With this final representation
the total quantities (\ref{tot1}) and (\ref{tot2}) look
formally like sums of two free-fields contributions. The very nonlocal
transformation (\ref{chph}) has now accommodated the mixing aspects of the
asymptotics.

It is interesting to note that $H(z)$ acquires here the role of a phase
in a very natural way. This is rather satisfying, since the same conclusion has
been reached earlier in a different way, by considering a quantum version of an
"adiabatic approximation", see ref. \cite{her} (for "quantum field" $\F(l)$
definition (\ref{phase}) gives $\F(g_z)$ of this reference).
Moreover, $-2H(z)$ is identical with the change of phase $\delta(z)$ (\ref{fasta})
in the external field problem calculated by Staruszkiewicz.
On the other hand $H(z)$ is distinct from a phase variable considered
by Staruszkiewicz \cite{sta89} in his theory of quantum Coulomb field.
We discuss the difference in some detail.
A phase field of \cite{sta89} is a homogeneous of degree $0$ field in the region
$x^2<0$, satisfying there the homogeneous wave equation. Such a field can be
represented by
$$
S(x)=\int\left\{\sgn(x\s l)\, f_1(l) +
\ln\frac{|x\s l|}{t\s l} f_2(l)\right\}\, \dl + c_t\, ,
$$
where $f_1$, $f_2$ are homogeneous of degree $-2$ functions of $l$ and
$\dsp \int f_2(l)\, \dl =0$, $t$ is a timelike, unit,
future-pointing vector and
$c_t$ is a constant; this constant changes for another choice of vector $t$
according to $\dsp c_{\tilde{t}}= c_t +
\int\ln\frac{\tilde{t}\s l}{t\s l}\, f_2(l)\, \dl$. Consider the spherically
symmetric term $S_z(x)$ in the expansion of $S(x)$ in spherical harmonics
in a coordinate system in which $z$ points in the direction of the
time-axis. One easily shows that
$$
S_z(x) = \int f_1(l)\, \dl\, \frac{x\s z}{\sqrt{(x\s z)^2 - x^2}} + c_z\, .
$$
Identifications of Staruszkiewicz are
$$
-\frac{1}{e}\int f_1(l)\, \dl = \text{charge}\, ,~~~~~
c_z = \text{phase~variable}\, .
$$
To compare this with our identifications we
use the relation of $S(x)$ to the long-range field of \cite{sta89}
$$
-eF^{\mrm{l.r.}}{}_{ab}(x) x^b = \n_aS(x)\, .
$$
One easily shows using (\ref{lE}) and (\ref{Ka}) that in our description
a field $S(x)$ which can be formed out of the long-range variables and which
satisfies this relation is given by
$$
S(x)=-\frac{e}{2\pi}\int\sgn(x\s l)\, (q(l)+\si(l))\, \dl\, .
$$
For this field
$$
S_z(x)= -e Q\frac{x\s z}{\sqrt{(x\s z)^2 - x^2}}\, .
$$
The charge part agrees with that of Staruszkiewicz, but the analog of his phase
variable is absent. Our phase variable, which is the null spherical harmonic in
$e\F(l)$, does not appear in $S_z(x)$. (The absence of logarithmic terms in
our version of $S(x)$ is due to the conditions on null asymptotics of fields.)

\section{Conclusions}

The main results of our analysis can be summarized as follows.
\begin{itemize}[topsep=1ex,itemsep=0ex,leftmargin=2.2em]
\item[(i)] Despite nonintegrability of the angular momentum tensor density over a Cauchy surface,
the total angular momentum (four dimensional) can be unambiguously identified, provided (a)
angular momentum radiated (or incoming from null directions) over finite time intervals is well
defined, and (b) the magnetic part of the spacelike asymptotic of the electromagnetic field vanishes.
\item[(ii)] Asymptotic Dirac field can be identified by a special choice of gauge and consideration of
the asymptotic behavior of the Dirac field on the hyperboloid $x^2=\lambda^2$ for $\lambda\to\infty$.
\item[(iii)] The total energy momentum of the system can be expressed as a sum of independent
contributions from the asymptotic free electromagnetic field and the asymptotic Dirac field. However,
in the analogous representation of the angular momentum an additional term survives, which
mixes the asymptotic Dirac field characteristic with the infrared characteristic of the free asymptotic
electromagnetic field. This effect persists in the limit of the energy tending to zero. The
additional term can be accommodated into the matter part by a redefinition of the asymptotic Dirac
field. This is a very nonlocal transformation mixing the matter aspects with the spacelike asymptotics
of radiation.
\end{itemize}

\section*{Acknowledgements}
\label{ack}
I would like to thank Professor D.\,Buchholz for careful reading of
the ma\-nu\-script and interesting discussions, and Professor
A.\,Staruszkiewicz for some remarks. I am grateful to the II. Institut
f\"{u}r Theoretische Physik, Universit\"{a}t Hamburg for hospitality and
to the Humboldt Foundation for financial support.

\setcounter{equation}{0}
\setcounter{pr}{0}
\setcounter{section}{0}
\renewcommand{\thesection}{Appendix \Alph{section}}
\renewcommand{\theequation}{\Alph{section}.\arabic{equation}}
\renewcommand{\thepr}{\Alph{section}.\arabic{pr}}
\section{Homogeneous functions\newline of a spinor variable}
\label{apa}

We reproduce here some facts about the invariant measure over the null
directions in Minkowski space \cite{sta,zwa,pen} and on spin-weighted spherical
harmonics \cite{pen}.

Let $u$ denote a vector on the future lightcone. The measure $\dsp
{d^3 u\over u^0}$ is known to be Lorentz invariant. If we define
a measure over the null directions
$d^2u$ by $\dsp {d^3 u\over u^0} = {du^0\over u^0} d^2u$
(in the notation of \cite{sta}), then the new measure is Lorentz invariant
in the following sense: The result of integration of a homogeneous of degree
$-2$ function of $u$ is manifestly Lorentz invariant.

The invariant measure has a very simple and elegant representation in the
spinor language \cite{pen}. If $\xi^A$ is a spinor of the null vector $u$,
then
$$
d^2u = i\xi^{A'} d\xi_{A'}\wedge \xi^A d\xi_A \, .
$$
Here any parametrization of spinors is implied for which every null direction
is represented by exactly one spinor. The scaling behaviour of $d^2 u$ is now
explicit:
$$
\text{if}~\xi\to\alpha\xi\, ,~\text{then}~\du\to(\alpha\ac)^2\du\, .
$$

Some special scalings of spinors are useful. We say that a spinor $o^A$ is
chosen in a $t$-gauge, if its null vector $l$ satisfies $t\s l =1$, where
$t$ is a~fixed unit timelike vector. In this scaling the measure $\du$ is
the rotationally invariant measure over unit $2$-sphere in the reference
system with time-axis along $t$ \cite{zwa,sta}, which we denote $d\W_t(u)$.

Let us choose a fixed spinor $o^A$ in a $t$-gauge and denote $\iota^A=
t^{AA'} o_{A'}$. Then $\{o^A, \iota^A\}$ is a normalised spinor basis:
$o_A\iota^A =1$. Parametrize $\xi^A$ with complex numbers $\kappa$ from the
closed unit circle by the formula
\be
\xi^A = (1-|\kappa|^2)^{1/2} o^A + \sqrt{2}\kappa\,\iota^A \, .
\label{par1}
\ee
Then $\xi^A$ is in a $t$-gauge and
$d\W_t(u) = 2i\, d\bar{\kappa}\wedge d\kappa$. Setting $\kappa=
\rho e^{\textstyle -i\varphi}$, $\rho\in\langle 0, 1\rangle$,
$\varphi\in\langle 0, 2\pi)$ we obtain $d\W_t(u) = 2d\rho^2\wedge d\varphi$.
Finally, substituting $\dsp \rho = \sin\frac{\vartheta}{2}$,
$\vartheta\in\langle 0,
\pi\rangle$ we obtain the spherical angles parametrization:
\be
\xi^A = \cos{\vartheta\over 2} o^A + \sqrt{2}\sin{\vartheta\over 2}
e^{\textstyle -i\varphi}\iota^A \, ,
\label{par2}
\ee
\be
u^a = t^a + \sin\vartheta (\cos\varphi X^a + \sin\varphi Y^a ) +
\cos\vartheta Z^a \, ,
\label{par3}
\ee
where $\dsp X^a = {1\over\sqrt{2}}(o^A\iota^{A'} + \iota^A o^{A'})$,
$\dsp Y^a = {i\over\sqrt{2}}(o^A\iota^{A'} - \iota^A o^{A'})$,
$Z^a = l^a - t^a$ is a~Cartesian basis, and $d\W_t(u) = \sin\vartheta\,
d\vartheta\wedge d\varphi$.

The invariant integral is an important tool in the theory of homogeneous
functions of a spinor variable, known as the theory of spin-weighted
spherical harmonics \cite{pen}. We reproduce some results of the theory needed
in the present paper.

A function $f(o, \oc)$ is said to be of type $\{p, q\}$ if
\be
f(\alpha o, \ac\oc)= \alpha^p\ac^q f(o, \oc)\, .
\label{swh1}
\ee
Choose any timelike versor $t^a$ and
denote $\dsp \io^A =\frac{t^{AA'}o_{A'}}{t\s l}$. Denote also
$\dsp \D_A=\frac{\D}{\D o^A}$,\linebreak
$\dsp \D_{A'}=\frac{\D}{\D o^{A'}}$. Then one has
\be
\begin{array}{@{}l}
\dsp \mrm{(i)~if}~p-q>0~\mrm{then}~\io^A\D_A f=0 \Rightarrow f=0;\\
\dsp \mrm{(ii)~if}~p-q<0~\mrm{then}~\io^{A'}\D_{A'} f=0
\Rightarrow f=0\, .
\end{array}
\label{swh2}
\ee
For $f_1$: $\{0, q_1\}$, $f_2$: $\{p_2, 0\}$ one has by Euler theorem
\be
\D_A f_1 = o_A g_1,~~~\D_{A'} f_2 = o_{A'} g_2\, .
\label{swh3}
\ee
If $q_1<0$, $p_2<0$ then
(i) and (ii) imply that $f_{\mrm{i}}$ are uniquely determined by
$g_{\mrm{i}}$ ($\mrm{i}=1,2$).
Moreover, if $q_1=p_2=-2$, so that $g_1$ and $g_2$ are of type $\{-2, -2\}$,
then
\be
\int g_1\, \dl = \int g_2\, \dl = 0\, .
\label{inth1}
\ee
Using this, one also easily shows that
\be
\int \D_A h_1\, \dl=0\, ,~~~\int \D_{A'} h_2\, \dl=0\, ,
\label{inth2}
\ee
for $h_1$: $\{-1, -2\}$, and $h_2$: $\{-2, -1\}$.

\setcounter{equation}{0}
\setcounter{pr}{0}
\section{Some estimates and limits}
\label{apb}

We prove here various estimates of asymptotic behaviour of fields and
quantities appearing in this paper. Our tool is the following simple
lemma. Let $a>0$, $b\geq 0$, $c>0$, $\alpha >0$ (all real).
Then
\be
\int_0^c (a + bu)^{-\alpha}\, du
< \left\{ \begin{array}{ll}{\dsp {\alpha \over \alpha -1}\,
{c\over a^{\alpha -1}(a + cb)}\, ,} &\alpha >1\\
{\dsp {1\over 1-\alpha}\,{c\over (a + cb)^\alpha} \, ,} &\alpha <1\, .
\end{array}\right. ~~~{}
\label{est}
\ee
To see this, represent the result of integration by
$\dsp \frac{c}{(a+bc)^\alpha}\, h_\alpha\lp\frac{cb}{a}\rp$ for $\alpha<1$ and by
$\dsp \frac{c}{a^{\alpha-1}(a+bc)}\, h_{2-\alpha}\lp\frac{cb}{a}\rp$ for
$\alpha>1$, where $\dsp h_\beta(x)=\frac{1}{(1-\beta)x}\lb x+1 - (x+1)^\beta\rb$
for $\beta<1$. $h_\beta(x)$ is a continuous function of $x\geq 0$, $h_\beta(0)=1$,
$\dsp \lim_{x\to\infty} h_\beta(x) = \frac{1}{1-\beta}$. Moreover, it is
easy to see that $\dsp x^2 h'_\beta(x)=\beta\int_0^x y(y+1)^{\beta-2}\, dy$,
so $h_\beta(x)$ is monotonous. Hence
$\dsp h_\beta(x)<\max\left\{1, \frac{1}{1-\beta}\right\}$, which ends
the proof of (\ref{est}).

Let $C_{\mrm{k}}(x)$, $\mrm{k}=0, 1, 2$, be free fields
$$
C_{\mrm{k}}(x) = -\frac{1}{2\pi}\int f_{\mrm{k}}(x\s u, u)\, \du\, ,
$$
such that for $s>0$ there is $\dsp |f_{\mrm{k}}(s, u)|<
\frac{\con}{(s+s_t)^{\mrm{k}+\e}}$ in some $t$-gauge ($s_t>0$). Then using
the above lemma we obtain immediately in $t$-gauge
\begin{eqnarray}
&&|C_0(st+Rl)|<\frac{\con}{(s+s_t+2R)^\e}\, ,
\label{bs0}\\
&&|C_1(st+Rl)|<\frac{\con}{(s+s_t)^\e(s+s_t+2R)}\, ,
\label{bs1}\\
&&|C_2(st+Rl)|<\frac{\con}{(s+s_t)^{1+\e}(s+s_t+2R)}\, .
\label{bs2}
\end{eqnarray}
For $x^a$ inside the forward lightcone set $x^a=\la z^a$, with $z^2=1$,
$z^0>0$. The above bounds imply then
\begin{eqnarray}
&&|C_0(\la z)|<\frac{\con}{(s_t +\la z^0)^\e}\, ,
\label{bl0}\\
&&|C_1(\la z)|<\con\frac{(z^0)^\e}{(\la +s_t z^0)^\e(s_t+\la z^0)}\, ,
\label{bl1}\\
&&|C_2(\la z)|<\con\frac{(z^0)^{1+\e}}{(\la +s_t z^0)^{1+\e}(s_t+\la z^0)}\, .
\label{bl2}
\end{eqnarray}
All bounds (\ref{bs0}-\ref{bl2}) hold also in other reference systems, with
$s_t$ and other constants depending on vector $t$.

We turn now to the energy-momentum and angular momentum of electromagnetic
field passing through the cone $\Co^{\mrm{fut}}(\tau)$:
\begin{eqnarray}
&&P^{\mrm{elm}}{}_a\lb\Co^{\mrm{fut}}(\tau)\rb=\int T^{\mrm{elm}}{}_{ac}
(\tau t+Rl)\, l^c\, R^2dR\, d\W_t(l)\, ,
\label{estcone1}\\
&&\mu^{\mrm{elm}}{}_{AB}\lb\Co^{\mrm{fut}}(\tau)\rb=\int
\mu^{\mrm{elm}}{}_{ABc}
(\tau t+Rl)\, l^c\, R^2dR\, d\W_t(l)\, ,
\label{estcone2}
\end{eqnarray}
where $l^c$ is in a $t$-gauge, $T^{\mrm{elm}}{}_{ac}$ and
$\mu^{\mrm{elm}}{}_{ABc}$ are given by (\ref{emt}) and (\ref{ams3})
respectively. We shall show that both quantities vanish in the limit
$\tau\to\infty$ for the free field (\ref{frel1}) and for the
mixed terms of the free
outgoing field and the advanced field of the asymptotic Dirac current.
To this end estimates of $\f_{AC}(\tau t+Rl)o^C$ and
$\fa_{A'C}(\tau t+Rl)o^C$ are needed. For the
free field we use the representations
(\ref{frel1}), (\ref{frel2}) with the spinor variable $\xi^A$ and integrate
by the use of identity
$$
\lp\frac{\tau}{2}+R\rp o_A\xi^A \dot{f}((\tau t+Rl)\s u, \xi, \xc)
= \io^{A'}\lp\frac{\D}{\D\xi^{A'}} - \frac{\D'}{\D\xi^{A'}}\rp
f((\tau t+Rl)\s u, \xi, \xc) \, ,
$$
($\io^{A'}=t^{A'A}o_A$). Then
\begin{eqnarray*}
&&\f_{AC}(\tau t+Rl)o^C =-\frac{1}{\pi(\tau +2R)}\int\xi_A\io^{B'}
\D'_{B'} \dot{\z}((\tau t+Rl)\s u, \xi, \xc)\, \du\, ,\\
&&\fa_{A'C}(\tau t+Rl)o^C =-\frac{1}{\pi(\tau +2R)}\int\io^{B'}
\D'_{B'}\D'_{A'}\z((\tau t+Rl)\s u, \xi, \xc)\, \du\, .
\end{eqnarray*}
The estimates (\ref{bs1}), (\ref{bs0}) give
\begin{eqnarray}
&&\left|\f_{AC}(\tau t+Rl)o^C \right|<\frac{\con}{\tau^\e(\tau+2R)^2}\, ,
\label{estcone3}\\
&&\left|\fa_{A'C}(\tau t+Rl)o^C \right|<\frac{\con}{(\tau+2R)^{1+\e}}\, .
\label{estcone4}
\end{eqnarray}
\begin{flushleft}
Using these bounds in (\ref{estcone1}) and (\ref{estcone2}) one gets
$\dsp \left|P_a\lb\Co^{\mrm{fut}}(\tau)\rb\right|<
\frac{\con}{\tau^{1+2\e}}$,
$\dsp \left|\mu_{AB}\lb\Co^{\mrm{fut}}(\tau)\rb\right|<
\frac{\con}{\tau^{2\e}}$ for the free field, which proves (\ref{freefut}).
For the advanced field we use (\ref{rest2})--(\ref{co2}). This yields
\end{flushleft}
\begin{eqnarray}
&&\left|\f^{\mrm{adv}}{}_{AB}(\tau t+Rl)\right|<\frac{\con}{(\tau+R)^2}\, ,
\label{estcone5}\\
&&\left|\lp\fa^{\mrm{adv}}{}_{AA'}-\fa^{\mrm{Coul}}{}_{AA'}\rp(\tau t+Rl)
\right|<\frac{\con}{(\tau+R)^{1+\e}}\, .
\label{estcone6}
\end{eqnarray}
For estimation of $\fa^{\mrm{Coul}}{}_{A'C}(\tau t+Rl)o^C$ one has to use
more specific algebraic property of the Coulomb field. One shows with the use
of (\ref{co1},\ref{co2}) that
$$
\f^{\mrm{Coul}}{}_{AB}(x)x^A_{A'}x^B_{B'} = \frac{x^2}{2}\int\rho(v)
\frac{v_{C(A'}x^C_{B')}}{\sqrt{(x\s v)^2-x^2}}\,
\frac{\m(v)}{(x\s v)^2-x^2}\, .
$$
The integral on the r.h. side is estimated as the Coulomb field itself.
Taking into account that
$$
\f^{\mrm{Coul}}{}_{AB}(x)x^A_{A'}x^B_{B'}\io^{B'}
\Big|_{\textstyle x=\tau t+Rl}
=\left(\frac{\tau}{2}+R\right)\fa^{\mrm{Coul}}{}_{A'C}(\tau t+Rl)o^C
$$
we get
\be
\left|\fa^{\mrm{Coul}}{}_{A'C}(\tau t+Rl)o^C\right|<\con
\frac{\tau}{(\tau +R)^2}\, .
\label{estcone7}
\ee
A straightforward calculation shows now that
(\ref{estcone5}--\ref{estcone7}) together with
(\ref{estcone3}, \ref{estcone4}) are
sufficient for vanishing of the mixed terms contributions to
(\ref{estcone1}, \ref{estcone2}) in the limit $\tau\to\infty$.

The last point in this appendix is the demonstration of (\ref{reste}) and
(\ref{resta}) for the mixed contributions of the free field and
the (generalized) Coulomb
field (\ref{Coul}) to the electromagnetic energy-momentum tensor. We give
explicit calculations for the case of $\Co'^{\mrm{fut}}(-r)$, where
\begin{eqnarray*}
&&P^{\mrm{mix}}{}_a\lb\Co'^{\mrm{fut}}(-r)\rb = \int d\W_t(u)
\int_r^\infty T^{\mrm{mix}}{}_{ac}(-rt+Ru)u^c\, R^2\, dR\, ,\\
&&\mu^{\mrm{mix}}{}_{AB}\lb\Co'^{\mrm{fut}}(-r)\rb = \int d\W_t(u)
\int_r^\infty \mu^{\mrm{mix}}{}_{ABc}(-rt+Ru)u^c\, R^2\, dR\, .
\end{eqnarray*}
We use (\ref{Coul}) (with $a=0$) and (\ref{frel1}, \ref{frel2}) in
mixed terms of (\ref{emt}) and (\ref{ams3}). Integration over $R$ is then
explicitly carried out with the use of identities
\begin{eqnarray*}
&&\f_{AC}(-rt+Ru)\xi^C = \D_R\lp-\frac{1}{2\pi}\rp
\int o_A\dot{\z}(-r+Rl\s u, o, \oc)\, \frac{d\W_t(l)}{o_{C'}\xi^{C'}}\, ,\\
&&\fa_{A'C}(-rt+Ru)\xi^C = \D_R\lp-\frac{1}{2\pi}\rp
\int \D'_{A'}\z(-r+Rl\s u, o, \oc)\, \frac{d\W_t(l)}{o_{C'}\xi^{C'}}\, .
\end{eqnarray*}
We get
\begin{align*}
 &P^{\mrm{mix}}{}_a\lb\Co'^{\mrm{fut}}(-r)\rb = -\frac{\ov{Q}}{8\pi^2}
 \int d\W_t(u)\, d\W_t(l) \frac{o_A\xi_{A'}}{o_{C'}\xi^{C'}}
 \dot{\z}(r(l\s u-1), o, \oc)+\cc\, ,\\
 &\mu^{\mrm{mix}}{}_{AB}\lb\Co'^{\mrm{fut}}(-r)\rb = \frac{1}{8\pi^2}
 \int d\W_t(u)\, d\W_t(l)\times\\
 &\hspace*{4em}\times\left\{\frac{Q}{o_C\xi^C}\xi_{(A}\D'_{B)}
 \ov{\z(r(l\s u-1), o, \oc)} -
 \frac{\ov{Q}r}{o_{C'}\xi^{C'}}\xi_{D'}t^{D'}_{(A} o_{B)}\dot{\z}(r(l\s u-1), o, \oc)\right\}\, .
\end{align*}
The energy-momentum expression vanishes in the limit $r\to\infty$ by the
Lebes\-gue theorem, while the first term in the angular momentum expression
yields
\begin{eqnarray*}
&&\frac{Q}{8\pi^2}\int d\W_t(l)\, d\W_t(u)\,
\frac{\theta(1-l\s u)}{o_C\xi^C} \xi_{(A}\D_{B)}\ov{\z(-\infty, o, \oc)}\\
&&=\frac{1}{4\pi}\int \z^Q{}_{(A}\ov{\nu}_{B)}(-\infty, o, \oc)\, \dl\, ,
\end{eqnarray*}
where $\xi$-integration in the parametrization (\ref{par2}) was performed. The
second term in this parametrization after $\f$-integration is
\begin{eqnarray*}
&&\frac{\ov{Q}r}{8\pi}\int d\W_t(l)\, d\theta\, \sin\vartheta o_Ao_B
\dot{\z}(-r\cos\theta, o, \oc)\\
&&\hspace*{-1em}=\frac{\ov{Q}}{8\pi}\int d\W_t(l)\, o_Ao_B\lb\z(r, o, \oc)-
\z(-r, o, \oc)\rb
\to\frac{1}{4\pi}\int\z_{(A}\ov{\nu}^Q{}_{B)}(-\infty, o, \oc)\, \dl
\end{eqnarray*}
for $r\to\infty$, which ends the proof.

\setcounter{equation}{0}
\setcounter{pr}{0}
\section{3-space integrals}
\label{apc}

We prove here a lemma, from which the formulae (\ref{emcf}) and (\ref{amcf})
for conserved quantities of a free electromagnetic field follow by a
simple computation.

\begin{lem}
\label{lemsi}
Let $f_1(s, o, \oc)$ and $f_2(s, o, \oc)$ be continuously differentiable
functions satisfying scaling law \mbox{$f(\alpha\ac s, \alpha o, \ac\oc) =
\alpha^{-2}\ac{}^{-1} f(s, o, \oc)$}, such that $|f_2(s, o, \oc)|$ is boun\-ded,
$|f_1(s, o, \oc)|$, $|\D_Af_1(s, o, \oc)|$ and $|\D_{A'}f_1(s, o, \oc)|$
are boun\-ded by an integrable function (in any fixed gauge), and there exist
limits  $\dsp \lim_{s\to+\infty}f_2(s, o, \oc) =
-\lim_{s\to-\infty}f_2(s, o, \oc)$.\\
Then
\begin{equation}\label{calsi}
\begin{aligned}
 \lim_{r\to\infty}&
 \int\limits_{\{ t\cdot x=c,~(t\cdot x)^2-x^2\leq r^2\} }\left\{
 t^a\frac{1}{2\pi}\int\ov{\dot{f}_1(x\s u, \xi, \xc)}\xi_{A'}\, \du
 \frac{1}{2\pi}\int \dot{f}_2(x\s l, o, \oc)o_A\, \dl\,\right\} d^3x\\[1ex]
 &=\int\ov{f_1(s, o, \oc)}f_2(s, o, \oc)\, ds\, \dl\, .
\end{aligned}
\end{equation}
\end{lem}
Note that the r.h. side is explicitly hyperplane-independent.
\begin{proof}
Choose a $t$-gauge $t\s l=t\s u=1$, fix $o^A$ and parametrize $\xi^A$
by (\ref{par2}) and $x^a$ by
\[
 x^a=x^0 t^a -y_1Z^a -y_2\lp\cos\f\, X^a+\sin\f\, Y^a\rp-
y_3\lp\sin\f\, X^a - \cos\f\, Y^a\rp\, .
\]
Then $x\s l=x^0+y_1$, $x\s u=x^0+y_2\sin\theta+y_1\cos\theta$. Hence
\begin{eqnarray*}
&&\int_{\{ x^0=c,~|\vec{x}|\leq r\} }
t^a\int\ov{\dot{f}_1(x\s u, \xi, \xc)}\xi_{A'}\, \du\,
\dot{f}_2(x\s l, o, \oc)o_A\, d^3x \\
&&=2\int d\Omega(\theta,\f)\int_{\left\{y_1^2+y_2^2\leq r^2\right\}}d^2y\,
\sqrt{r^2-y_1^2-y_2^2}\times\\
&&\hspace*{3cm}\times\cos\frac{\theta}{2}
\ov{\dot{f}_1(y_1\cos\theta+y_2\sin\vartheta + c, \xi, \xc)} \dot{f}_2(y_1 + c, o, \oc)
\\
&&=2\int d\theta\, d\f \int_{\left\{y_1^2+y_2^2\leq r^2\right\}}d^2y\,
\frac{y_2}{\sqrt{r^2-y_1^2-y_2^2}}\times \\
&&\hspace*{3cm}\times\cos\frac{\theta}{2}
\ov{f_1(y_1\cos\theta+y_2\sin\vartheta +c, \xi, \xc)} \dot{f}_2(y_1 +c, o, \oc)\, .
\end{eqnarray*}
The effect of the constant $c$ is a translation of both functions in the first
argument, so if the lemma is proved for $c=0$, then it is true for all $c$.
We set $c=0$ for simplicity.
By the change of variables $s=y_1\cos\vartheta +y_2\sin\theta$,
$v=-y_1\sin\theta+y_2\cos\theta$ we get
\begin{align*}
 &-2\int d\theta\, d\f\int_{\left\{s^2+v^2\leq r^2\right\}}\frac{ds\, dv}
 {\sqrt{r^2-s^2-v^2}}\, \cos\frac{\theta}{2}\ov{f_1(s, \xi, \xc)}
 \D_{\theta} f_2(s\cos\theta-v\sin\theta, o, \oc)\\
 &=-2\int_{-r}^r ds\int_{-1}^1 \frac{d\kappa}{\sqrt{1-\kappa^2}}
 \int d\theta\, d\f \cos\frac{\theta}{2}\ov{f_1(s, \xi, \xc)}
 \D_{\theta} f_2(s\cos\theta-\sqrt{r^2-s^2}\kappa\sin\theta, o, \oc)\, .
\end{align*}
Integrating  by parts over $\theta$ we obtain
\begin{align*}
 &(2\pi)^2\int_{-r}^r \ov{f_1(s, o, \oc)} f_2(s, o, \oc)\, ds
 +2\int_{-r}^r ds\int_{0}^1 \frac{d\kappa}{\sqrt{1-\kappa^2}}\int d\theta\,d\f\,
 \D_{\theta}\lp\cos\frac{\theta}{2}\ov{f_1(s, \xi, \xc)}\rp \times\\[1ex]
 &\times\lb f_2(s\cos\theta+\sqrt{r^2-s^2}\kappa\sin\theta, o, \oc) +
 f_2(s\cos\theta-\sqrt{r^2-s^2}\kappa\sin\theta, o, \oc)\rb\, .
\end{align*}
By the Lebesgue theorem the second integral vanishes in the limit, which
ends the proof.
\end{proof}

\setcounter{equation}{0}
\setcounter{pr}{0}
\section{An estimate}
\label{apd}

We show here, that if $2\geq\beta\geq 0 $, $\g\geq 0$, $\alpha>|\beta+\g -1|$ and
$\dsp |G(z, v)|<\frac{\con}{(v^0)^{\alpha +1}}$ for both
$z$ and $v$ on the unit four-velocity hyperboloid, then
\be
\left|\int\frac{G(z, v)}{\lp\sqrt{(z\s v)^2 -1}\rp^\beta
\lp z\s v + \sqrt{(z\s v)^2 -1}\rp^\g}\, \m(v)\right|
< \frac{\con}{(z^0)^{\beta+\g}}\, .
\label{mest}
\ee
The bound is then valid in any other reference system (with some other
constant).

Let $t^a$ be the time-axis versor of the reference system. We show that
$$
I\equiv\int\frac{\m(v)}{\lp\sqrt{(z\s v)^2 -1}\rp^\beta
\lp z\s v + \sqrt{(z\s v)^2 -1}\rp^\g (t\s v)^{\alpha +1}} <
\frac{\con}{(t\s z)^{\beta+\g}}\, .
$$
Choose the time-axis of the coordinate system in which integration is performed
along $z$ and set $v^0=\sqrt{|\vec{v}|^2 +1}$. Then
$$
I=2\pi\int\frac{|\vec{v}|^{2-\beta}\, d|\vec{v}|}{v^0 (v^0+|\vec{v}|)^\g}\,
\int_0^2\frac{d\xi}{\lp\cosh\ch v^0 - \sinh\ch |\vec{v}|
+ \sinh\ch |\vec{v}|\xi\rp^{\alpha +1}}\, ,
$$
where $\cosh\ch=t\s z$, $\ch\geq 0$. By (\ref{est}) the inside integral
is bounded by
\[
 \frac{\con}{(\cosh\ch v^0 - \sinh\ch |\vec{v}|)^\alpha
(\cosh\ch v^0 + \sinh\ch |\vec{v}|)}.
\]
By change of integration variable
$|\vec{v}| = \sinh\psi$ we get
\begin{align*}
 &I<\con\int_0^\infty\frac{(\sinh\psi)^{2-\beta}\, d\psi}{(e^{\textstyle\psi})^\g
 \cosh(\psi+\ch)(\cosh(\psi-\ch))^\alpha} \\
 &<\frac{\con}{(\cosh\ch)^{\beta+\g}}\int_{-\ch}^{\infty}
\frac{d\psi}{(e^{\textstyle\psi})^{\beta+\g-1}(\cosh\psi)^\alpha}
 < \frac{\con}{(\cosh\ch)^{\beta+\g}}\int_{-\infty}^{+\infty}
\frac{d\psi}{(e^{\textstyle\psi})^{\beta+\g-1}(\cosh\psi)^\alpha}\, .
\end{align*}

\newpage

\end{document}